\documentclass[11pt]{article}%
\usepackage{amsfonts}
\usepackage{amsmath}
\usepackage{fullpage}
\usepackage{amssymb}
\usepackage{graphicx}
\usepackage{hyperref}%
\providecommand{\U}[1]{\protect\rule{.1in}{.1in}}
\newtheorem{theorem}{Theorem}

\newtheorem{corollary}[theorem]{Corollary}

\newtheorem{definition}[theorem]{Definition}

\newtheorem{lemma}[theorem]{Lemma}

\newtheorem{proposition}[theorem]{Proposition}
\newtheorem{remark}[theorem]{Remark}

\newenvironment{proof}[1][Proof]{\noindent\textbf{#1.} }{\ \rule{0.5em}{0.5em}}
\numberwithin{equation}{section}

\begin{document}

\title{\textbf{Quantum Markov chains, sufficiency of quantum channels, and R\'{e}nyi
information measures}}
\author{Nilanjana Datta\thanks{Statistical Laboratory, Centre for Mathematical
Sciences, University of Cambridge, Wilberforce Road, Cambridge CB3 0WB, UK}
\and Mark M. Wilde\thanks{Hearne Institute for Theoretical Physics, Department of
Physics and Astronomy, Center for Computation and Technology, Louisiana State
University, Baton Rouge, Louisiana 70803, USA}}
\maketitle

\begin{abstract}
A short quantum Markov chain is a tripartite state $\rho_{ABC}$ such that
system $A$ can be recovered perfectly by acting on system $C$ of the reduced
state $\rho_{BC}$. Such states have conditional mutual information
$I(A;B|C)$\ equal to zero and are the only states with this property. A
quantum channel $\mathcal{N}$\ is sufficient for two states $\rho$ and
$\sigma$ if there exists a recovery channel using which one can perfectly
recover $\rho$ from $\mathcal{N}(\rho)$ and $\sigma$ from $\mathcal{N}%
(\sigma)$. The relative entropy difference $D(\rho\Vert\sigma)-D(\mathcal{N}%
(\rho)\Vert\mathcal{N}(\sigma))$ is equal to zero if and only if $\mathcal{N}%
$\ is sufficient for $\rho$ and $\sigma$. In this paper, we show that these
properties extend to R\'{e}nyi generalizations of these information measures
which were proposed in [Berta \textit{et al.}, J.~Math.~Phys.~\textbf{56},
022205, (2015) and Seshadreesan \textit{et al.}, J.~Phys.~A \textbf{48},
395303, 2015], thus providing an alternate characterization of short quantum
Markov chains and sufficient quantum channels. These results give further
support to these quantities as being legitimate R\'{e}nyi generalizations of
the conditional mutual information and the relative entropy difference. Along
the way, we solve some open questions of Ruskai and Zhang, regarding the trace
of particular matrices that arise in the study of monotonicity of relative
entropy under quantum operations and strong subadditivity of the von Neumann entropy.

\end{abstract}

\section{Introduction}

Markov chains and sufficient statistics are two fundamental notions in
probability \cite{Feller,Norris} and statistics \cite{stats}. Three random
variables $X$, $Y$, and $Z$ constitute a three-step Markov chain (denoted as
$X-Y-Z$) if $X$ and $Z$ are independent when conditioned on $Y$. In
particular, if $p_{XYZ}(x,y,z)$ is their joint probability distribution, then%
\begin{align}
p_{XYZ}(x,y,z)  &  =p_{X}(x)\ p_{Y|X}(y|x)\ p_{Z|Y}(z|y)\nonumber\\
&  =p_{X|Y}(x|y)\ p_{Z|Y}(z|y)\ p_{Y}(y). \label{eq:MC}%
\end{align}
In the information-theoretic framework, such a Markov chain corresponds to a
\emph{{recoverability condition}} in the following sense. Consider $X$, $Y$,
and $Z$ to be the inputs and outputs of two channels (i.e., stochastic maps)
$p_{Y|X}$ and $p_{Z|Y}$, as in the figure below. \begin{figure}[ptb]
\begin{center}
\includegraphics[
width=3.00in
]
{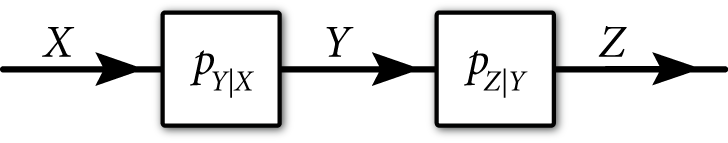}
\end{center}
\end{figure}If $X-Y-Z$ is a three-step Markov chain, then the input $X$, if
lost, can be recovered from $Y$ alone (without any knowledge of $Z$) by the
action of the stochastic map $p_{X|Y}$, as is evident from (\ref{eq:MC}).

It is well known that such a Markov chain $X-Y-Z$ can be characterized by an
information measure \cite{book1991cover}, namely, the \emph{{conditional
mutual information}} $I(X;Z|Y)$. For any three random variables $X$, $Y$ and
$Z$, it is defined as%
\begin{equation}
I(X;Z|Y)\equiv H(XY)+H(ZY)-H(Y)-H(XYZ),
\end{equation}
where $H(W)\equiv-\sum_{w}p_{W}(w)\log p_{W}(w)$ is the Shannon entropy of a
random variable $W\sim p_{W}(w)$. It is non-negative and equal to zero if and
only if $X-Y-Z$ is a Markov chain.

In statistics, for a given sample of independent and identically distributed
data conditioned on an unknown parameter $\theta$, a \emph{{sufficient
statistic}} is a function of the sample whose value contains all the
information needed to compute any estimate of the parameter. One can extend
this notion to that of a sufficient channel (or sufficient stochastic map), as
discussed in \cite{Petz1986,Petz1988,Mosonyi2004,M05}. A channel $T\equiv
T_{Y|X}$ is sufficient for two input distributions $p_{X}$ and $q_{X}$ if
there exists another channel (a \emph{{recovery channel}}) such that both
these inputs can be recovered perfectly by sending the outputs of the channel
$T_{Y|X}$ corresponding to them through it. This notion of channel sufficiency
is likewise characterized by an information measure, namely, the relative
entropy difference%
\begin{equation}
D(p_{X}\Vert q_{X})-D(T(p_{X})\Vert T(q_{X})), \label{eq:cl-rel-ent-diff}%
\end{equation}
where $T(p_{X})$ and $T(q_{X})$ are the distributions obtained after the
action of the channel, and $D(p_{X}\Vert q_{X})$ denotes the relative entropy
(or Kullback-Leibler divergence)~\cite{book1991cover} between $p_{X}$ and
$q_{X}$. It is defined as%
\begin{equation}
D(p_{X}\Vert q_{X})\equiv\sum_{x}p_{X}(x)\log\left(  \frac{p_{X}(x)}{q_{X}%
(x)}\right)  ,
\end{equation}
if for all $x$, $q_{X}(x)\neq0$ if $p_{X}(x)\neq0$. It is equal to $+\infty$
otherwise. The notion of recoverability provides a connection between the
notions of Markov chains and sufficient channels.

The generalization of the above ideas to quantum information theory has been a
topic of continuing and increasing interest (see,
e.g.,~\cite{HMPB11,BSW14,FR14} and references therein). In the quantum
setting, density operators play a role analogous to that of probability
distributions in the classical case, and in \cite{HJPW04}, a quantum Markov
chain $A-C-B$\ was defined to be a tripartite density operator $\rho_{ABC}$
with conditional (quantum) mutual information $I(A;B|C)_{\rho}$ equal to zero,
where%
\begin{equation}
I(A;B|C)_{\rho}\equiv H(AC)_{\rho}+H(BC)_{\rho}-H(C)_{\rho}-H(ABC)_{\rho},
\label{eq:CMI}%
\end{equation}
and $H(F)_{\sigma}\equiv-$Tr$\left\{  \sigma_{F}\log\sigma_{F}\right\}  $
denotes the von Neumann entropy of a density operator $\sigma_{F}$. (We take
the convention $A-C-B$ for a quantum Markov chain because we are often
interested in quantum correlations between Alice ($A$) and Bob ($B$), which
are potentially mediated by a third party, here labeled by $C$.) Strong
subadditivity of the von Neumann entropy guarantees that the conditional
mutual information $I(A;B|C)_{\rho}$ is non-negative for all density operators
\cite{PhysRevLett.30.434,LR73}, and it is equal to zero if and only if there
is a decomposition of $\mathcal{H}_{C}$ as%
\begin{equation}
\mathcal{H}_{C}=\bigoplus\limits_{j}\mathcal{H}_{C_{L_{j}}}\otimes
\mathcal{H}_{C_{R_{j}}} \label{eq:h-c-direct-sum}%
\end{equation}
such that%
\begin{equation}
\rho_{ABC}=\bigoplus\limits_{j}q(j)\rho_{AC_{L_{j}}}^{j}\otimes\rho_{C_{R_{j}%
}B}^{j}, \label{eq:direct-sum-structure}%
\end{equation}
for a probability distribution $\{q(j)\}$ and sets of density operators
$\{\rho_{AC_{L_{j}}}^{j},\rho_{C_{R_{j}}B}^{j}\}$ \cite{HJPW04}. Following
\cite{HJPW04}, we call such states \emph{short quantum Markov chains} $A-C-B$.
In analogy with the classical case, $I(A;B|C)_{\rho}=0$ is equivalent to the
full state $\rho_{ABC}$ being recoverable after the loss of system $A$ by the
action of a quantum recovery channel $\mathcal{R}_{C\rightarrow AC}^{P}$\ on
system $C$ alone:%
\begin{equation}
I(A;B|C)_{\rho}=0\Leftrightarrow\rho_{ABC}=\mathcal{R}_{C\rightarrow AC}%
^{P}(\rho_{BC}),
\end{equation}
where%
\begin{equation}
\mathcal{R}_{C\rightarrow AC}^{P}(\cdot)\equiv\rho_{AC}^{\frac{1}{2}}\rho
_{C}^{-\frac{1}{2}}(\cdot)\rho_{C}^{-\frac{1}{2}}\rho_{AC}^{\frac{1}{2}}
\label{eq:special-Petz}%
\end{equation}
is a special case of the so-called \emph{Petz recovery channel}%
~\cite{Petz1986,Petz1988}. Note that this channel acts as the identity on
system $B$ and the superscript $P$ refers to Petz.

As a generalization of the classical notion of a sufficient channel, several
works have discussed and studied the notion of a sufficient quantum channel
\cite{Petz1986,Petz1988,Mosonyi2004,M05}. A definition of this concept is as follows:

\begin{definition}
[Sufficiency of a quantum channel]\label{def:channel-suff}Let $\rho$ and
$\sigma$ be density operators acting on a Hilbert space $\mathcal{H}$, and let
$\mathcal{N}$ be a quantum channel taking these density operators to density
operators acting on a Hilbert space $\mathcal{K}$. Then the quantum channel
$\mathcal{N}$\ is sufficient for them if one can perfectly recover $\rho$ from
$\mathcal{N(}\rho)$ and $\sigma$ from $\mathcal{N(}\sigma)$ by the action of a
quantum recovery channel $\mathcal{R}$, i.e., if there exists an $\mathcal{R}$
such that%
\begin{equation}
\rho=\left(  \mathcal{R\circ N}\right)  (\rho),\ \ \ \ \ \ \ \ \sigma=\left(
\mathcal{R\circ N}\right)  (\sigma). \label{eq:recover-rho}%
\end{equation}
We define sufficiency in the same way even if $\rho$ and $\sigma$ are
arbitrary positive semi-definite operators.
\end{definition}

If \eqref{eq:recover-rho} is true for some recovery channel $\mathcal{R}$, it
is known that the following Petz recovery channel $\mathcal{R}_{\sigma
,\mathcal{N}}^{P}$ satisfies (\ref{eq:recover-rho}) as well \cite{Petz1988}:
\begin{equation}
\mathcal{R}_{\sigma,\mathcal{N}}^{P}(\omega)\equiv\sigma^{\frac{1}{2}%
}\mathcal{N}^{\dag}\left(  [\mathcal{N}(\sigma)]^{-\frac{1}{2}}\omega
\lbrack\mathcal{N}(\sigma)]^{-\frac{1}{2}}\right)  \sigma^{\frac{1}{2}}.
\label{eq:Petz-map}%
\end{equation}
(Note that the Petz recovery channel in (\ref{eq:special-Petz}) is a special
case of (\ref{eq:Petz-map}) with $\sigma=\rho_{AC}$ and $\mathcal{N}%
=\operatorname{Tr}_{A}$.) As a generalization of the classical case, the
sufficiency of a quantum channel is characterized by the following information
measure, the relative entropy difference%
\begin{equation}
D(\rho\Vert\sigma)-D(\mathcal{N(}\rho)\Vert\mathcal{N(}\sigma)),
\label{eq:rel-ent-diff}%
\end{equation}
where $D(\rho\Vert\sigma)$ denotes the \emph{{quantum relative entropy}}
\cite{U62}. It is defined as%
\begin{equation}
D(\rho\Vert\sigma)\equiv\text{Tr}\left\{  \rho\left[  \log\rho-\log
\sigma\right]  \right\}  , \label{rel-entropy}%
\end{equation}
whenever the support of $\rho$ is contained in the support of $\sigma$ and it
is equal to $+\infty$ otherwise. The relative entropy difference in
(\ref{eq:rel-ent-diff}) is non-negative due to the monotonicity of relative
entropy under quantum channels \cite{Lindblad1975,U77}, and it is equal to
zero if and only if $\mathcal{N}$ is sufficient for $\rho$ and $\sigma$, i.e.,
if the Petz recovery channel $\mathcal{R}_{\sigma,\mathcal{N}}^{P}$ satisfies
(\ref{eq:recover-rho}) \cite{Petz1986,Petz1988}. Further, Mosonyi and Petz
have shown that the relative entropy difference in (\ref{eq:rel-ent-diff}) is
equal to zero if and only if $\rho$, $\sigma$, and $\mathcal{N}$ have the
explicit form recalled below in Theorem~\ref{thm:mos-petz-decomp} of
Section~\ref{sec:prelim} \cite{Mosonyi2004,M05}, which generalizes the result
stated in (\ref{eq:h-c-direct-sum})-(\ref{eq:direct-sum-structure}).

Due to its operational interpretation in the quantum Stein's lemma
\cite{HP91}, the quantum relative entropy plays a central role in quantum
information theory. In particular, fundamental limits on the performance of
information-processing tasks in the so-called \textquotedblleft asymptotic,
memoryless (or i.i.d.) setting\textquotedblright\ are given in terms of
quantities derived from the quantum relative entropy.

There are, however, other generalized relative entropies (or divergences)
which are also of operational significance. Important among these are the
\emph{R\'{e}nyi relative entropies}~\cite{Renyi,P86} and the more recently
defined \emph{sandwiched R\'{e}nyi relative entropies}~\cite{MDSFT13, WWY13}.
For $\alpha\in(0,1)$, the R\'{e}nyi relative entropies arise in the quantum
Chernoff bound \cite{chernoff}, which characterizes the minimum probability of
error in discriminating two different quantum states in the setting of
asymptotically many copies. Moreover, in analogy with the operational
interpretation of their classical counterparts, the R\'{e}nyi relative
entropies can be viewed as generalized cutoff rates in quantum binary state
discrimination~\cite{milan}. The sandwiched R\'{e}nyi relative entropies find
application in the strong converse domain of a number of settings dealing with
hypothesis testing or channel capacity \cite{WWY13,MO13,GW13,TWW14,CMW14,HT14}.

This motivates the introduction of R\'{e}nyi generalizations of the
conditional mutual information. Two of these generalizations, defined in
\cite{BSW14}, are given as follows:%
\begin{align}
I_{\alpha}(A;B|C)_{\rho}  &  \equiv\frac{1}{\alpha-1}\log\text{Tr}\left\{
\rho_{ABC}^{\alpha}\rho_{AC}^{\frac{1-\alpha}{2}}\rho_{C}^{\frac{\alpha-1}{2}%
}\rho_{BC}^{1-\alpha}\rho_{C}^{\frac{\alpha-1}{2}}\rho_{AC}^{\frac{1-\alpha
}{2}}\right\}  ,\label{eq:Renyi-CMI}\\
\widetilde{I}_{\alpha}(A;B|C)_{\rho}  &  \equiv\frac{2\alpha}{\alpha-1}%
\log\left\Vert \rho_{ABC}^{\frac{1}{2}}\rho_{AC}^{\frac{1-\alpha}{2\alpha}%
}\rho_{C}^{\frac{\alpha-1}{2\alpha}}\rho_{BC}^{\frac{1-\alpha}{2\alpha}%
}\right\Vert _{2\alpha}, \label{eq:sand-Renyi-cmi}%
\end{align}
where $\alpha\in\left(  0,1\right)  \cup\left(  1,\infty\right)  $ denotes the
R\'{e}nyi parameter. (Note that we use the notation $\left\Vert A\right\Vert
_{\alpha}\equiv\operatorname{Tr}\{(\sqrt{A^{\dag}A})^{\alpha}\}^{1/\alpha}$
even for $\alpha\in\left(  0,1\right)  $, when it is not a norm.) Both these
quantities converge to (\ref{eq:CMI}) in the limit $\alpha\rightarrow1$, they
are non-negative,$\,$ and obey several properties of the conditional mutual
information defined in (\ref{eq:CMI}), as shown in \cite{BSW14}. In
\cite{SBW14}, the authors proposed some definitions for R\'{e}nyi
generalizations of a relative entropy difference, two of which are as follows:%
\begin{align}
\Delta_{\alpha}(\rho,\sigma,\mathcal{N})  &  \equiv\frac{1}{\alpha-1}%
\log\text{Tr}\left\{  \rho^{\alpha}\sigma^{\frac{1-\alpha}{2}}\mathcal{N}%
^{\dag}\left(  \mathcal{N}(\sigma)^{\frac{\alpha-1}{2}}\mathcal{N}%
(\rho)^{1-\alpha}\mathcal{N}(\sigma)^{\frac{\alpha-1}{2}}\right)
\sigma^{\frac{1-\alpha}{2}}\right\}  ,\label{eq:Renyi-rel-ent-diff}\\
\widetilde{\Delta}_{\alpha}(\rho,\sigma,\mathcal{N})  &  \equiv\frac{\alpha
}{\alpha-1}\log\left\Vert \rho^{\frac{1}{2}}\sigma^{\frac{1-\alpha}{2\alpha}%
}\mathcal{N}^{\dag}\left(  \mathcal{N}(\sigma)^{\frac{\alpha-1}{2\alpha}%
}\mathcal{N(}\rho)^{\frac{1-\alpha}{\alpha}}\mathcal{N(}\sigma)^{\frac
{\alpha-1}{2\alpha}}\right)  \sigma^{\frac{1-\alpha}{2\alpha}}\rho^{\frac
{1}{2}}\right\Vert _{\alpha}. \label{eq:Renyi-rel-ent-diff-sand}%
\end{align}
The quantities above converge to (\ref{eq:rel-ent-diff}) in the limit
$\alpha\rightarrow1$ \cite{SBW14}.

The quantities defined in (\ref{eq:Renyi-CMI}%
)-(\ref{eq:Renyi-rel-ent-diff-sand}) can be expressed in terms of R\'{e}nyi
relative entropies, namely the $\alpha$-R\'{e}nyi relative entropy and the
$\alpha$-sandwiched R\'{e}nyi relative entropy defined in
Section~\ref{Sec-rel-entropies}~\cite{BSW14}. The corresponding expressions
for the R\'{e}nyi generalizations of the conditional mutual information are%
\begin{align}
I_{\alpha}(A;B|C)_{\rho}  &  =D_{\alpha}\left(  \rho_{ABC}\middle\Vert\left(
\rho_{AC}^{\frac{1-\alpha}{2}}\rho_{C}^{\frac{\alpha-1}{2}}\rho_{BC}%
^{1-\alpha}\rho_{C}^{\frac{\alpha-1}{2}}\rho_{AC}^{\frac{1-\alpha}{2}}\right)
^{\frac{1}{1-\alpha}}\right)  ,\\
\widetilde{I}_{\alpha}(A;B|C)_{\rho}  &  =\widetilde{D}_{\alpha}\left(
\rho_{ABC}\middle\Vert\left(  \rho_{AC}^{\frac{1-\alpha}{2\alpha}}\rho
_{C}^{\frac{\alpha-1}{2\alpha}}\rho_{BC}^{\frac{1-\alpha}{\alpha}}\rho
_{C}^{\frac{\alpha-1}{2\alpha}}\rho_{AC}^{\frac{1-\alpha}{2\alpha}}\right)
^{\frac{\alpha}{1-\alpha}}\right)  ,
\end{align}
respectively, and those for the relative entropy difference are given in
(\ref{eq:ren-rel-ent-diff-1}) and (\ref{eq:ren-rel-ent-diff-2}), respectively.

Before proceeding, we should note that the conditional mutual information is a
special case of a relative entropy difference. Indeed, one can check that%
\begin{equation}
I(A;B|C)_{\rho}=D(\rho\Vert\sigma)-D(\mathcal{N}(\rho)\Vert\mathcal{N}%
(\sigma)) \label{eq:rel-ent-cmi}%
\end{equation}
for the choices%
\begin{equation}
\rho=\rho_{ABC},\ \ \ \ \ \ \ \ \sigma=\rho_{AC}\otimes I_{B}%
,\ \ \ \ \ \ \ \ \mathcal{N}=\operatorname{Tr}_{A}. \label{eq:rel-ent-cmi-1}%
\end{equation}
This reduction extends as well to the R\'{e}nyi quantities for all $\alpha
\in(0,1)\cup(1,\infty)$:%
\begin{align}
I_{\alpha}(A;B|C)_{\rho}  &  =\Delta_{\alpha}(\rho_{ABC},\rho_{AC}\otimes
I_{B},\operatorname{Tr}_{A}),\\
\widetilde{I}_{\alpha}(A;B|C)_{\rho}  &  =\widetilde{\Delta}_{\alpha}%
(\rho_{ABC},\rho_{AC}\otimes I_{B},\operatorname{Tr}_{A}),
\end{align}
as pointed out in \cite{SBW14}. This realization is helpful in simplifying
some of the arguments in this paper.

\section{Summary of results}

As highlighted in the Introduction, an important property of the conditional
(quantum) mutual information of a tripartite quatum state is that it is always
non-negative and vanishes if and only if the state is a short quantum Markov
chain. The relative entropy difference of a pair of quantum states and a
quantum channel is also non-negative and vanishes if and only if the channel
is sufficient for the pair of states. Consequently, it is reasonable to
require that R\'{e}nyi generalizations of these information measures are also
non-negative and vanish under the same necessary and sufficient conditions as
mentioned above.

In this paper, we prove these properties for the quantities $I_{\alpha}$ and
$\widetilde{I}_{\alpha}$, and the quantities $\Delta_{\alpha}$ and
$\widetilde{\Delta}_{\alpha}$ defined in the Introduction. This contributes
further evidence that $I_{\alpha}( A;B|C) _{\rho}$ and $\widetilde{I}_{\alpha
}( A;B|C) _{\rho}$ are legitimate R\'{e}nyi generalizations of the conditional
mutual information $I( A;B|C) _{\rho}$ and that $\Delta_{\alpha}( \rho
,\sigma,\mathcal{N}) $ and $\widetilde{\Delta}_{\alpha}( \rho,\sigma
,\mathcal{N}) $ are legitimate R\'{e}nyi generalizations of the relative
entropy difference $D( \rho\Vert\sigma) -D( \mathcal{N}( \rho) \Vert
\mathcal{N}( \sigma) ) $. In particular, we prove the following:

\begin{enumerate}
\item $I_{\alpha}( A;B|C) _{\rho}=0$ for some $\alpha\in\left(  0,1\right)
\cup\left(  1,2\right)  $ if and only if $\rho_{ABC}$ is a short quantum
Markov chain with a decomposition as in (\ref{eq:direct-sum-structure}).

\item $\widetilde{I}_{\alpha}( A;B|C) _{\rho}=0$ for some $\alpha\in\left(
1/2,1\right)  \cup\left(  1,\infty\right)  $ if and only if $\rho_{ABC}$ is a
short quantum Markov chain with a decomposition as in
(\ref{eq:direct-sum-structure}).

\item $\Delta_{\alpha}( \rho,\sigma,\mathcal{N}) $ is non-negative for
$\alpha\in\left(  0,1\right)  \cup\left(  1,2\right)  $ and $\widetilde
{\Delta}_{\alpha}( \rho,\sigma,\mathcal{N}) $ is non-negative for $\alpha
\in\left(  1/2,1\right)  \cup\left(  1,\infty\right)  $.

\item $\Delta_{\alpha}( \rho,\sigma,\mathcal{N}) =0$ for some $\alpha
\in\left(  0,1\right)  \cup\left(  1,2\right)  $ if and only if the quantum
channel $\mathcal{N}$ is sufficient for states $\rho$ and $\sigma$, so that
$\mathcal{N}$, $\rho$, and $\sigma$ decompose as in
(\ref{eq:sufficiency-states-decompose}%
)-(\ref{eq:sufficiency-channel-decompose}).

\item $\widetilde{\Delta}_{\alpha}( \rho,\sigma,\mathcal{N}) =0$ for some
$\alpha\in\left(  1/2,1\right)  \cup\left(  1,\infty\right)  $ if and only if
the quantum channel $\mathcal{N}$ is sufficient for states $\rho$ and $\sigma
$, so that $\mathcal{N}$, $\rho$, and $\sigma$ decompose as in
(\ref{eq:sufficiency-states-decompose}%
)-(\ref{eq:sufficiency-channel-decompose}).

\item Generalizations of the conditional mutual information and the relative
entropy difference arising from the so-called min- and max-relative
entropies~\cite{min-max-ND, Dupuis} (which play an important role in one-shot
information theory) satisfy identical properties.
\end{enumerate}

Along the way, we resolve some open questions stated in \cite{R02,Z14b}. Let
$\rho_{ABC}$ be a positive definite density operator. We prove that%
\begin{equation}
\operatorname{Tr}\left\{  \left(  \rho_{AC}^{\frac{1-\alpha}{2}}\rho
_{C}^{\frac{\alpha-1}{2}}\rho_{BC}^{1-\alpha}\rho_{C}^{\frac{\alpha-1}{2}}%
\rho_{AC}^{\frac{1-\alpha}{2}}\right)  ^{\frac{1}{1-\alpha}}\right\}  \leq1.
\end{equation}
for all $\alpha\in\left(  0,1\right)  \cup\left(  1,2\right)  $, and%
\begin{equation}
\operatorname{Tr}\left\{  \left(  \rho_{AC}^{\frac{1-\alpha}{2\alpha}}\rho
_{C}^{\frac{\alpha-1}{2\alpha}}\rho_{BC}^{\frac{1-\alpha}{\alpha}}\rho
_{C}^{\frac{\alpha-1}{2\alpha}}\rho_{AC}^{\frac{1-\alpha}{2\alpha}}\right)
^{\frac{\alpha}{1-\alpha}}\right\}  \leq1,
\end{equation}
for all $\alpha\in\left(  1/2,1\right)  \cup\left(  1,\infty\right)  $. Let
$\rho$ and $\sigma$ be positive definite density operators and let
$\mathcal{N}$ be a strict completely positive trace preserving (CPTP) map
(that is, a CPTP map such that $\mathcal{N}(X)$ is positive definite whenever
$X$ is positive definite). We prove that%
\begin{equation}
\operatorname{Tr}\left\{  \left[  \sigma^{\frac{1-\alpha}{2}}\mathcal{N}%
^{\dag}\left(  \mathcal{N}(\sigma)^{\frac{\alpha-1}{2}}\mathcal{N}%
(\rho)^{1-\alpha}\mathcal{N}(\sigma)^{\frac{\alpha-1}{2}}\right)
\sigma^{\frac{1-\alpha}{2}}\right]  ^{\frac{1}{1-\alpha}}\right\}  \leq1,
\label{eq:N-rho-sigma-leq-1}%
\end{equation}
for $\alpha\in\left(  0,1\right)  \cup(1,2)$, and%
\begin{equation}
\operatorname{Tr}\left\{  \left[  \sigma^{\frac{1-\alpha}{2\alpha}}%
\mathcal{N}^{\dag}\left(  \mathcal{N}(\sigma)^{\frac{\alpha-1}{2\alpha}%
}\mathcal{N}(\rho)^{\frac{1-\alpha}{\alpha}}\mathcal{N}(\sigma)^{\frac
{\alpha-1}{2\alpha}}\right)  \sigma^{\frac{1-\alpha}{2\alpha}}\right]
^{\frac{\alpha}{1-\alpha}}\right\}  \leq1, \label{eq:2-N-rho-sigma-leq-1}%
\end{equation}
for $\alpha\in\left(  1/2,1\right)  \cup(1,\infty)$. By taking the limit
$\alpha\rightarrow1$, the inequalities in \eqref{eq:N-rho-sigma-leq-1} and
\eqref{eq:2-N-rho-sigma-leq-1} imply that%
\begin{equation}
\operatorname{Tr}\left\{  \exp\left\{  \log\sigma+\mathcal{N}^{\dag}\left(
\log\mathcal{N}(\rho)-\log\mathcal{N}(\sigma)\right)  \right\}  \right\}
\leq1.
\end{equation}

The rest of the paper is devoted to establishing these claims. We begin by
recalling some mathematical preliminaries and known results, and follow by
establishing the latter claims first and then move on to the former ones.

\section{Preliminaries}

\label{sec:prelim} Let $\mathcal{B}(\mathcal{H})$ denote the algebra of
bounded linear operators acting on a Hilbert space $\mathcal{H}$. We restrict
to finite-dimensional Hilbert spaces throughout this paper. We denote the
support of an operator~$A$ by supp$(A)$. For a Hermitian operator $A$, by
$A^{-1}$ we mean the inverse restricted to supp$(A)$, so that $AA^{-1}%
=A^{-1}A$ is the orthogonal projection onto supp$(A)$. More generally, for a
function~$f$ and Hermitian operator $A$ with spectral decomposition
$A=\sum_{i}\lambda_{i}|i\rangle\langle i|$, we define $f(A)$ to be
$\sum_{i:\lambda_{i}\neq0}f(\lambda_{i})|i\rangle\langle i|$. Let
$\mathcal{B}(\mathcal{H})_{+}$ denote the subset of positive semidefinite
operators, and let $\mathcal{B}(\mathcal{H})_{++}$ denote the subset of
positive definite operators. We also write $X\geq0$ if $X\in\mathcal{B}%
(\mathcal{H})_{+}$ and $X>0$ if $X\in\mathcal{B}(\mathcal{H})_{++}$. An
operator $\rho$ is in the set $\mathcal{S}(\mathcal{H})$ of density operators
(or states) if $\rho\in\mathcal{B}(\mathcal{H})_{+}$ and Tr$\{\rho\}=1$, and
an\ operator $\rho$ is in the set $\mathcal{S}(\mathcal{H})_{++}$\ of positive
definite density operators if $\rho\in\mathcal{B}(\mathcal{H})_{++}$ and
Tr$\{\rho\}=1$. Throughout much of the paper, for technical convenience and
simplicity, we consider states in $\mathcal{S}(\mathcal{H})_{++}$. For
$\alpha\geq1$, we define the $\alpha$-norm of an operator $X$ as%
\begin{equation}
\left\Vert X\right\Vert _{\alpha}\equiv\left[  \text{Tr}\{|X|^{\alpha
}\}\right]  ^{1/\alpha}, \label{eq:a-norm}%
\end{equation}
where $|X|\equiv\sqrt{X^{\dag}X}$, and we use the same notation even for the
case $\alpha\in(0,1)$, when it is not a norm. The fidelity~$F(\rho,\sigma)$ of
two states $\rho,\sigma\in\mathcal{S}(\mathcal{H})$ is defined as
\begin{equation}
F(\rho,\sigma)\equiv\left\Vert \sqrt{\rho}\sqrt{\sigma}\right\Vert _{1}^{2}.
\end{equation}

A quantum channel is given by a completely positive, trace-preserving (CPTP)
map $\mathcal{N}:\mathcal{B}(\mathcal{H})\mapsto\mathcal{B}(\mathcal{K})$,
with $\mathcal{H}$ and $\mathcal{K}$ being the input and output Hilbert spaces
of the channel, respectively. Let $\left\langle C,D\right\rangle
\equiv\operatorname{Tr}\left\{  C^{\dag}D\right\}  $ denote the
Hilbert-Schmidt inner product of $C,D\in\mathcal{B}(\mathcal{H})$. The adjoint
of the quantum channel $\mathcal{N}$ is a completely positive unital map
$\mathcal{N}^{\dagger}:\mathcal{B}(\mathcal{K})\mapsto\mathcal{B}%
(\mathcal{H})$ defined through the following relation:%
\begin{equation}
\left\langle B,\mathcal{N}(A)\right\rangle =\langle\mathcal{N}^{\dagger
}(B),A\rangle,
\end{equation}
for all $A\in\mathcal{B}(\mathcal{H)}$ and $B\in\mathcal{B}(\mathcal{K})$. A
linear map is said to be a strict CPTP map if it is CPTP and if $\mathcal{N}%
(A)\in\mathcal{B}(\mathcal{K})_{++}$ for all $A\in\mathcal{B}(\mathcal{H}%
)_{++}$. Note that a CPTP\ map is strict if and only if $\mathcal{N}(I)>0$
\cite[Section~2.2]{B07}. We denote the identity channel as id but often
suppress it for notational simplicity.

The set $\left\{  U^{i}\right\}  $ of Heisenberg-Weyl unitaries acting on a
finite-dimensional Hilbert space $\mathcal{H}$ of dimension $d$ has the
property that%
\begin{equation}
\frac{1}{d^{2}}\sum_{i}U^{i}X\left(  U^{i}\right)  ^{\dag}=\text{Tr}\left\{
X\right\}  \frac{I}{d},
\end{equation}
for any operator $X$ acting on $\mathcal{H}$.

\subsection{Generalized relative entropies}

\label{Sec-rel-entropies} The following relative entropies of a density
operator $\rho$ with respect to a positive semidefinite operator $\sigma$ play
an important role in this paper. In what follows, we restrict the definitions
to the case in which $\rho$ and $\sigma$ satisfy the condition
$\operatorname{supp}\rho\subseteq\operatorname{supp}\sigma$, with the
understanding that they are equal to $+\infty$ when $\alpha>1$ and
$\operatorname{supp}\rho\not \subseteq \operatorname{supp}\sigma$. The
$\alpha$-R\'{e}nyi relative entropy is defined for $\alpha\in\left(
0,1\right)  \cup\left(  1,\infty\right)  $ as follows \cite{P86}:%
\begin{equation}
D_{\alpha}(\rho\Vert\sigma)\equiv\frac{1}{\alpha-1}\log\text{Tr}\left\{
\rho^{\alpha}\sigma^{1-\alpha}\right\}  .
\end{equation}
The $\alpha$-sandwiched R\'{e}nyi relative entropy is defined for $\alpha
\in\left(  0,1\right)  \cup\left(  1,\infty\right)  $ as follows
\cite{MDSFT13,WWY13}:%
\begin{equation}
\widetilde{D}_{\alpha}(\rho\Vert\sigma)\equiv\frac{1}{\alpha-1}\log
\text{Tr}\left\{  \left(  \sigma^{\frac{1-\alpha}{2\alpha}}\rho\sigma
^{\frac{1-\alpha}{2\alpha}}\right)  ^{\alpha}\right\}  .
\end{equation}
Both quantities above reduce to the quantum relative entropy in
(\ref{rel-entropy}) in the limit$~\alpha\rightarrow1$.

A fundamental property of the quantum relative entropy is that it is monotone
with respect to quantum channels (also known as the \emph{{data-processing
inequality}}):%
\begin{equation}
D(\rho\Vert\sigma)\geq D\left(  \mathcal{N}(\rho)\Vert\mathcal{N}%
(\sigma)\right)  , \label{DPI}%
\end{equation}
where $\mathcal{N}$ is a quantum channel. It is known that the data-processing
inequality is satisfied by the $\alpha$-R\'{e}nyi relative entropy for
$\alpha\in\lbrack0,1)\cup(1,2]$ {\cite{P86}}, and for the $\alpha$-sandwiched
R\'{e}nyi relative entropy for $\alpha\in\lbrack1/2,1)\cup(1,\infty]$
\cite{FL13,B13,MDSFT13,WWY13,MO13}.

Two special cases of the $\alpha$-sandwiched R\'{e}nyi relative entropy are of
particular significance in one-shot information theory~\cite{renato-thesis,
marco-thesis}, namely the min-relative entropy~\cite{Dupuis} and the
max-relative entropy~\cite{min-max-ND}. These are defined as follows:%
\begin{equation}
D_{\min}(\rho\Vert\sigma)\equiv\widetilde{D}_{1/2}(\rho\Vert\sigma)=-\log
F(\rho,\sigma),
\end{equation}
and
\begin{equation}
D_{\max}(\rho\Vert\sigma)\equiv\inf\{\lambda:\rho\leq2^{\lambda}\sigma
\}=\lim_{\alpha\rightarrow\infty}\widetilde{D}_{\alpha}(\rho\Vert\sigma).
\end{equation}
The relative entropies defined above, satisfy the following lemma
\cite{MDSFT13}:

\begin{lemma}
\label{lem:min-max-0}For $\omega\in\mathcal{S}( \mathcal{H}) $ and $\tau
\in\mathcal{B}( \mathcal{H}) _{+} $, such that $\operatorname{Tr}\left\{
\omega\right\}  \geq\operatorname{Tr}\left\{  \tau\right\}  $,%
\begin{align}
D_{\alpha}(\omega\Vert\tau)  &  \geq0\text{ for }\alpha\in\left(  0,1\right)
\cup\left(  1,2\right)  ,\\
\widetilde{D}_{\alpha}(\omega\Vert\tau)  &  \geq0\text{ for }\alpha\in\left(
1/2,1\right)  \cup\left(  1,\infty\right)  ,\\
D_{\min}(\omega\Vert\tau)  &  \geq0,\label{dmin-0}\\
D_{\max}(\omega\Vert\tau)  &  \geq0, \label{dmax-0}%
\end{align}
with equalities holding if and only if $\omega=\tau$.
\end{lemma}

In proving our results, we also employ the notion of a \emph{{quantum
$f$-divergence}}, first introduced by Petz~\cite{P86}. It can be conveniently
expressed as follows~\cite{FQAEP}:

\begin{definition}
For $A\in\mathcal{B}( \mathcal{H}) _{+}$, $B\in\mathcal{B}\left(
\mathcal{H}\right)  _{++}$ and an operator convex function $f$ on $[0,\infty
)$, the $f$-divergence of $A$ with respect to $B$ is given by
\begin{equation}
S_{f}(A\Vert B)=\langle\Gamma|\left(  \sqrt{B}\otimes I\right)  f\left(
B^{-1}\otimes A^{T}\right)  \left(  \sqrt{B}\otimes I\right)  |\Gamma\rangle,
\end{equation}
where $|\Gamma\rangle=\sum_{i}|i\rangle\otimes|i\rangle$, and $\{|i\rangle\}$
is an orthonormal basis of $\mathcal{H}$ with respect to which the transpose
is defined.
\end{definition}

\begin{remark}
\label{rem:renyi-f-div}Special cases of this are as follows:

\begin{enumerate}
\item The trace expression $\operatorname{Tr}\left\{  \rho^{\alpha}%
\sigma^{1-\alpha}\right\}  $ of the $\alpha$-R\'{e}nyi relative entropy, for
the choice $f(x)=x^{\alpha}$ for $\alpha\in\left(  1,2\right)  $, and
$-\operatorname{Tr}\left\{  \rho^{\alpha}\sigma^{1-\alpha}\right\}  $ for the
choice $f(x)=-x^{\alpha}$ for $\alpha\in\left(  0,1\right)  $.

\item The quantum relative entropy for the choice $f(x)=x\log x$.
\end{enumerate}
\end{remark}

The equivalence relations given in Lemma~\ref{lem-equiv} below follow directly
from \cite[Theorem~5.1]{HMPB11}.

\begin{lemma}
\label{lem-equiv} For $A,B\in\mathcal{B}(\mathcal{H})_{++}$ and a strict CPTP
map $\mathcal{N}$ acting on $\mathcal{B}\left(  \mathcal{H}\right)  $, the
following conditions are equivalent:{%
\begin{align}
S_{f}(\mathcal{N}(A)\Vert\mathcal{N}(B))  &  =S_{f}(A\Vert B)\quad
{\hbox{for the functions in Remark~\ref{rem:renyi-f-div}}},\label{item1}\\
{\mathcal{N}}^{\dagger}\left[  \log\mathcal{N}(A)-\log\mathcal{N}(B)\right]
&  =\log A-\log B,\label{item2}\\
&  \mathcal{N}\text{ is sufficient for }A\text{ and }B. \label{eq:item3}%
\end{align}
}
\end{lemma}

\begin{lemma}
\label{lem-MC} If $\rho_{ABC}$ is a positive definite density operator such
that
\begin{equation}
\label{eq-equal}\log\rho_{ABC}=\log\rho_{AC}+\log\rho_{BC}-\log\rho_{C},
\end{equation}
then it is a short quantum Markov chain $A-C-B$.
\end{lemma}

\begin{proof}
The identity (\ref{eq-equal}) is known from \cite{R02} to be a condition for
the conditional mutual information $I(A;B|C)_{\rho}$ of $\rho_{ABC}$ to be
equal to zero, which, by \cite{HJPW04}, implies that $\rho_{ABC}$ is a short
quantum Markov chain $A-C-B$.
\end{proof}

\begin{theorem}
[\cite{Mosonyi2004,M05}]\label{thm:mos-petz-decomp}Let $\rho\in\mathcal{S}(
\mathcal{H}) _{++}$, $\sigma\in\mathcal{B}(\mathcal{H})_{++}$, and
$\mathcal{N}$ be a strict CPTP\ map. Then $\mathcal{N}$ is sufficient for
$\rho$ and $\sigma$ (as in \eqref{eq:recover-rho}) if and only if the
following conditions hold

\begin{enumerate}
\item There exist decompositions of $\mathcal{H}$ and $\mathcal{K}$ as
follows:%
\begin{equation}
\mathcal{H}=\bigoplus\limits_{j}\mathcal{H}_{L_{j}}\otimes\mathcal{H}_{R_{j}%
},\ \ \ \ \ \ \ \ \mathcal{K}=\bigoplus\limits_{j}\mathcal{K}_{L_{j}}%
\otimes\mathcal{K}_{R_{j}}, \label{eq:suff-decomp}%
\end{equation}
where $\dim\left(  \mathcal{H}_{L_{j}}\right)  =\dim\left(  \mathcal{K}%
_{L_{j}}\right)  $ for all $j$.

\item With respect to the decomposition in \eqref{eq:suff-decomp}, $\rho$ and
$\sigma$ can be written as follows:%
\begin{equation}
\rho=\bigoplus\limits_{j}p( j) \rho_{L_{j}}^{j}\otimes\tau_{R_{j}}%
^{j},\ \ \ \ \ \ \ \ \sigma=\bigoplus\limits_{j}q( j) \sigma_{L_{j}}%
^{j}\otimes\tau_{R_{j}}^{j}, \label{eq:sufficiency-states-decompose}%
\end{equation}
for some probability distribution $\{p( j) \}$, positive reals $\{q( j) \}$,
sets of states $\{\rho_{L_{j}}^{j}\}$ and $\{\tau_{R_{j}}^{j}\}$ and set of
positive definite operators $\{\sigma_{L_{j}}^{j}\}$.

\item With respect to the decomposition in \eqref{eq:suff-decomp}, the quantum
channel $\mathcal{N}$ can be written as%
\begin{equation}
\mathcal{N}=\bigoplus\limits_{j}\mathcal{U}_{j}\otimes\mathcal{N}_{j}^{R},
\label{eq:sufficiency-channel-decompose}%
\end{equation}
where $\left\{  \mathcal{U}_{j}:\mathcal{B}(\mathcal{H}_{L_{j}})\rightarrow
\mathcal{B}(\mathcal{K}_{L_{j}})\right\}  $ is a set of unitary channels and
$\{\mathcal{N}_{j}^{R}:\mathcal{B}(\mathcal{H}_{R_{j}})\mapsto\mathcal{B}%
(\mathcal{K}_{R_{j}})\}$ is a set of quantum channels. Furthermore, with
respect to the decomposition in \eqref{eq:suff-decomp}, the adjoint of
$\mathcal{N}$ acts as%
\begin{equation}
\mathcal{N}^{\dag}=\bigoplus\limits_{j}\mathcal{U}_{j}^{\dag}\otimes\left(
\mathcal{N}_{j}^{R}\right)  ^{\dag}.
\end{equation}

\end{enumerate}
\end{theorem}

\subsection{Trace inequalities}

The following lemma is a consequence of \cite{epstein1973} (see also
\cite[Theorem~1.1]{Carlen2008}):

\begin{lemma}
\label{lem:concave-norm}For $A\in\mathcal{B}( \mathcal{H}) $, $B\in
\mathcal{B}( \mathcal{H}) _{+}$, and $p\in\left(  0,1\right)  $, the map
$B\mapsto\operatorname{Tr}\{\left(  AB^{p}A^{\dag}\right)  ^{1/p}\}$ is
concave. For invertible $A\in\mathcal{B}\left(  \mathcal{H}\right)  $,
$B\in\mathcal{B}( \mathcal{H}) _{++}$, and $p\in\left(  -1,0\right)  $, the
map $B\mapsto\operatorname{Tr}\{\left(  AB^{p}A^{\dag}\right)  ^{1/p}\}$ is concave.
\end{lemma}

\begin{lemma}
\label{lem:trace-ineq-gen}Let $\rho\in\mathcal{S}(\mathcal{H})$, $\sigma
\in\mathcal{B}(\mathcal{H})_{+}$, and let $\mathcal{N}:\mathcal{B}%
(\mathcal{H})\mapsto\mathcal{B}(\mathcal{K})$ be a CPTP\ map. For $\alpha
\in\left(  0,1\right)  $%
\begin{equation}
\operatorname{Tr}\left\{  \left[  \sigma^{\frac{1-\alpha}{2}}\mathcal{N}%
^{\dag}\left(  \mathcal{N}(\sigma)^{\frac{\alpha-1}{2}}\mathcal{N}%
(\rho)^{1-\alpha}\mathcal{N}(\sigma)^{\frac{\alpha-1}{2}}\right)
\sigma^{\frac{1-\alpha}{2}}\right]  ^{\frac{1}{1-\alpha}}\right\}  \leq1.
\label{eq:trace-ineq-rel-ent-diff}%
\end{equation}
For $\alpha\in\left(  1/2,1\right)  $%
\begin{equation}
\operatorname{Tr}\left\{  \left[  \sigma^{\frac{1-\alpha}{2\alpha}}%
\mathcal{N}^{\dag}\left(  \mathcal{N}(\sigma)^{\frac{\alpha-1}{2\alpha}%
}\mathcal{N}(\rho)^{\frac{1-\alpha}{\alpha}}\mathcal{N}(\sigma)^{\frac
{\alpha-1}{2\alpha}}\right)  \sigma^{\frac{1-\alpha}{2\alpha}}\right]
^{\frac{\alpha}{1-\alpha}}\right\}  \leq1. \label{eq:other-trace-ineq-petz}%
\end{equation}
Let $\rho\in\mathcal{S}(\mathcal{H})_{++}$, $\sigma\in\mathcal{B}%
(\mathcal{H})_{++}$, and let $\mathcal{N}:\mathcal{B}(\mathcal{H}%
)\mapsto\mathcal{B}(\mathcal{K})$ be a CPTP\ map such that $\mathcal{N}%
(\rho)\in\mathcal{S}(\mathcal{K})_{++}$, $\mathcal{N}(\sigma)\in
\mathcal{B}(\mathcal{K})_{++}$. For these choices, the inequality in
\eqref{eq:trace-ineq-rel-ent-diff} holds if $\alpha\in(1,2)$, and the
inequality in \eqref{eq:other-trace-ineq-petz} holds if $\alpha\in(1,\infty)$.
\end{lemma}

\begin{proof}
We begin by proving (\ref{eq:trace-ineq-rel-ent-diff}) for $\alpha\in\left(
0,1\right)  $. By Stinespring's dilation theorem~\cite{S55}, a given quantum
channel $\mathcal{N}$ can be realized as%
\begin{equation}
\mathcal{N}(\omega)=\text{Tr}_{E^{\prime}}\left\{  U\left(  \omega
\otimes|0\rangle\langle0|_{E}\right)  U^{\dag}\right\}  \quad\forall
\,\,\omega\in\mathcal{B}\left(  \mathcal{H}\right)  ,
\end{equation}
for some unitary $U$ taking $\mathcal{H\otimes H}_{E}$ to $\mathcal{K\otimes
H}_{E^{\prime}}$ and fixed state $|0\rangle_{E}$ in an auxiliary Hilbert
space$~{\mathcal{H}}_{E}$. Furthermore, it suffices to take $\dim\left(
\mathcal{H}_{E}\right)  \leq\dim(\mathcal{H})\dim\left(  \mathcal{K}\right)  $
because the number of Kraus operators for the channel $\mathcal{N}$\ can
always be taken less than or equal to $\dim(\mathcal{H})\dim\left(
\mathcal{K}\right)  $ and this is the dimension needed for an environment
system $\mathcal{H}_{E}$. The adjoint of $\mathcal{N}$ is given by%
\begin{equation}
\mathcal{N}^{\dag}\left(  \tau\right)  =\langle0|_{E}U^{\dag}\left(
\tau\otimes I_{E^{\prime}}\right)  U|0\rangle_{E}\quad\forall\,\,\tau
\in\mathcal{B}\left(  \mathcal{K}\right)  ,
\end{equation}
Then%
\begin{align}
&  \text{Tr}\left\{  \left(  \sigma^{\frac{1-\alpha}{2}}\mathcal{N}^{\dag
}\left(  \mathcal{N}(\sigma)^{\frac{\alpha-1}{2}}\mathcal{N}(\rho)^{1-\alpha
}\mathcal{N}(\sigma)^{\frac{\alpha-1}{2}}\right)  \sigma^{\frac{1-\alpha}{2}%
}\right)  ^{\frac{1}{1-\alpha}}\right\} \nonumber\\
&  =\text{Tr}\left\{  \left(  \sigma^{\frac{1-\alpha}{2}}\langle0|_{E}U^{\dag
}\left(  \left[  \left(  \mathcal{N}\left(  \sigma\right)  ^{\frac{\alpha
-1}{2}}\mathcal{N}(\rho)^{1-\alpha}\mathcal{N}(\sigma)^{\frac{\alpha-1}{2}%
}\right)  \right]  \otimes I_{E^{\prime}}\right)  U|0\rangle_{E}\sigma
^{\frac{1-\alpha}{2}}\right)  ^{\frac{1}{1-\alpha}}\right\} \nonumber\\
&  =\text{Tr}\left\{  \left(  AA^{\dag}\right)  ^{\frac{1}{1-\alpha}}\right\}
\end{align}
where%
\begin{align}
A  &  =\sigma^{\frac{1-\alpha}{2}}\langle0|_{E}U^{\dag}K_{\alpha}^{\dag},\\
K_{\alpha}  &  \equiv K_{\alpha}(\rho,\sigma,\mathcal{N})\equiv\mathcal{N}%
(\rho)^{\frac{1-\alpha}{2}}\mathcal{N}(\sigma)^{\frac{\alpha-1}{2}}\otimes
I_{E^{\prime}}. \label{eq:k_alpha}%
\end{align}
Then the above is equal to%
\begin{equation}
\text{Tr}\left\{  \left(  A^{\dag}A\right)  ^{\frac{1}{1-\alpha}}\right\}
=\text{Tr}\left\{  \left(  K_{\alpha}U|0\rangle_{E}\sigma^{\frac{1-\alpha}{2}%
}\sigma^{\frac{1-\alpha}{2}}\langle0|_{E}U^{\dag}K_{\alpha}^{\dag}\right)
^{\frac{1}{1-\alpha}}\right\}  , \label{eq:connect}%
\end{equation}
because the eigenvalues of $AA^{\dag}$ and $A^{\dag}A$ are the same for any
operator $A$. Using that%
\begin{align}
U|0\rangle_{E}\sigma^{\frac{1-\alpha}{2}}\sigma^{\frac{1-\alpha}{2}}%
\langle0|_{E}U^{\dag}  &  =U\left(  \sigma^{1-\alpha}\otimes|0\rangle
\langle0|_{E}\right)  U^{\dag}\nonumber\\
&  =U\left(  \sigma\otimes|0\rangle\langle0|_{E}\right)  ^{1-\alpha}U^{\dag
}\nonumber\\
&  =\left[  U\left(  \sigma\otimes|0\rangle\langle0|_{E}\right)  U^{\dag
}\right]  ^{1-\alpha},
\end{align}
the right hand side of (\ref{eq:connect}) is equal to%
\begin{equation}
\text{Tr}\left\{  \left(  K_{\alpha}\left[  U\left(  \sigma\otimes
|0\rangle\langle0|_{E}\right)  U^{\dag}\right]  ^{1-\alpha}K_{\alpha}^{\dag
}\right)  ^{\frac{1}{1-\alpha}}\right\}  . \label{eq:connect-up-2}%
\end{equation}
Let $\left\{  U_{E^{\prime}}^{i}\right\}  $ be a set of Heisenberg-Weyl
operators for the $E^{\prime}$ system and let $\pi_{E^{\prime}}$ denote the
maximally mixed state on system $E^{\prime}$. Now we use
Lemma~\ref{lem:concave-norm}\ to establish the inequality below:%
\begin{align}
&  \text{Tr}\left\{  \left(  K_{\alpha}\left[  U\left(  \sigma\otimes
|0\rangle\langle0|_{E}\right)  U^{\dag}\right]  ^{1-\alpha}K_{\alpha}^{\dag
}\right)  ^{\frac{1}{1-\alpha}}\right\} \nonumber\\
&  =\frac{1}{d_{E^{\prime}}^{2}}\sum_{i}\text{Tr}\left\{  \left(  K_{\alpha
}\left[  U_{E^{\prime}}^{i}U\left(  \sigma\otimes|0\rangle\langle
0|_{E}\right)  U^{\dag}\left(  U_{E^{\prime}}^{i}\right)  ^{\dag}\right]
^{1-\alpha}K_{\alpha}^{\dag}\right)  ^{\frac{1}{1-\alpha}}\right\} \\
&  \leq\text{Tr}\left\{  \left(  K_{\alpha}\left[  \frac{1}{d_{E^{\prime}}%
^{2}}\sum_{i}U_{E^{\prime}}^{i}U\left(  \sigma\otimes|0\rangle\langle
0|_{E}\right)  U^{\dag}\left(  U_{E^{\prime}}^{i}\right)  ^{\dag}\right]
^{1-\alpha}K_{\alpha}^{\dag}\right)  ^{\frac{1}{1-\alpha}}\right\} \\
&  =\text{Tr}\left\{  \left(  K_{\alpha}\left[  \text{Tr}_{E^{\prime}}\left\{
U\left(  \sigma\otimes|0\rangle\langle0|_{E}\right)  U^{\dag}\right\}
\otimes\pi_{E^{\prime}}\right]  ^{1-\alpha}K_{\alpha}^{\dag}\right)
^{\frac{1}{1-\alpha}}\right\}  .
\end{align}
Continuing, the last line above is equal to%
\begin{align}
&  \text{Tr}\left\{  \left(  K_{\alpha}\left[  \mathcal{N}(\sigma)\otimes
\pi_{E^{\prime}}\right]  ^{1-\alpha}K_{\alpha}^{\dag}\right)  ^{\frac
{1}{1-\alpha}}\right\} \nonumber\\
&  =\text{Tr}\left\{  \left(  K_{\alpha}\left[  \mathcal{N}(\sigma)^{1-\alpha
}\otimes\pi_{E^{\prime}}^{1-\alpha}\right]  K_{\alpha}^{\dag}\right)
^{\frac{1}{1-\alpha}}\right\} \label{eq:identities-1}\\
&  =\text{Tr}\left\{  \left(  \mathcal{N}(\rho)^{\frac{1-\alpha}{2}%
}\mathcal{N}(\sigma)^{\frac{\alpha-1}{2}}\mathcal{N}(\sigma)^{1-\alpha
}\mathcal{N}(\sigma)^{\frac{\alpha-1}{2}}\mathcal{N}(\rho)^{\frac{1-\alpha}%
{2}}\otimes\pi_{E^{\prime}}^{1-\alpha}\right)  ^{\frac{1}{1-\alpha}}\right\}
\\
&  =\text{Tr}\left\{  \left(  \mathcal{N}(\rho)^{\frac{1-\alpha}{2}%
}[\mathcal{N}(\sigma)]^{0}\mathcal{N}(\rho)^{\frac{1-\alpha}{2}}\otimes
\pi_{E^{\prime}}^{1-\alpha}\right)  ^{\frac{1}{1-\alpha}}\right\} \\
&  \leq\text{Tr}\left\{  \left(  \mathcal{N}(\rho)^{1-\alpha}\otimes
\pi_{E^{\prime}}^{1-\alpha}\right)  ^{\frac{1}{1-\alpha}}\right\} \\
&  =1, \label{eq:identities-2}%
\end{align}
where the inequality follows because%
\begin{equation}
\mathcal{N}(\rho)^{\frac{1-\alpha}{2}}\left[  \mathcal{N}(\sigma)\right]
^{0}\mathcal{N}(\rho)^{\frac{1-\alpha}{2}}\leq\mathcal{N}(\rho)^{1-\alpha},
\end{equation}
and%
\begin{equation}
\operatorname{Tr}\left\{  f(A)\right\}  \leq\operatorname{Tr}\left\{
f(B)\right\}  \label{eq:trace-monotone-1}%
\end{equation}
when $A\leq B$ and $f(x)\equiv x^{1/\left(  1-\alpha\right)  }$ is a monotone
non-decreasing function on $[0,\infty)$. The other inequality in
(\ref{eq:other-trace-ineq-petz}) follows by recognizing that%
\begin{equation}
\frac{1-\alpha}{\alpha}=1-\gamma, \label{eq:gamma-subs}%
\end{equation}
with $\gamma\equiv\left(  2\alpha-1\right)  /\alpha$, and $\gamma\in\left(
0,1\right)  $ when $\alpha\in(1/2,1)$.\ Thus, we can apply
(\ref{eq:trace-ineq-rel-ent-diff}) to conclude (\ref{eq:other-trace-ineq-petz}).

To prove the last two statements, we exploit Choi's inequality
\cite[Corollary~2.3]{choi1974} (see also \cite[Theorem~2.3.6]{B07}), which
states that%
\begin{equation}
\mathcal{M}(A)^{-1}\leq\mathcal{M}(A^{-1}),
\end{equation}
for $\mathcal{M}$ a strictly positive map and $A\in\mathcal{B}(\mathcal{H}%
)_{++}$.\ Consider that the adjoint $\mathcal{N}^{\dag}$\ of a channel
$\mathcal{N}$ is strictly positive because it is unital (recall that a map
$\mathcal{M}$\ is strictly positive if and only if $\mathcal{M}(I)>0$, which
is the case if $\mathcal{M}$ is unital). So let $\alpha\in(1,2)$ and consider
that%
\begin{multline}
\operatorname{Tr}\left\{  \left[  \sigma^{\frac{1-\alpha}{2}}\mathcal{N}%
^{\dag}\left(  \mathcal{N}(\sigma)^{\frac{\alpha-1}{2}}\mathcal{N}%
(\rho)^{1-\alpha}\mathcal{N}(\sigma)^{\frac{\alpha-1}{2}}\right)
\sigma^{\frac{1-\alpha}{2}}\right]  ^{\frac{1}{1-\alpha}}\right\} \\
=\operatorname{Tr}\left\{  \left[  \sigma^{-\frac{1-\beta}{2}}\mathcal{N}%
^{\dag}\left(  \mathcal{N}(\sigma)^{-\frac{\beta-1}{2}}\mathcal{N}%
(\rho)^{-\left(  1-\beta\right)  }\mathcal{N}(\sigma)^{-\frac{\beta-1}{2}%
}\right)  \sigma^{-\frac{1-\beta}{2}}\right]  ^{-\frac{1}{1-\beta}}\right\}  ,
\end{multline}
for $\beta=2-\alpha\in\left(  0,1\right)  $. Then%
\begin{align}
&  \left[  \sigma^{-\frac{1-\beta}{2}}\mathcal{N}^{\dag}\left(  \mathcal{N}%
(\sigma)^{-\frac{\beta-1}{2}}\mathcal{N}(\rho)^{-\left(  1-\beta\right)
}\mathcal{N}(\sigma)^{-\frac{\beta-1}{2}}\right)  \sigma^{-\frac{1-\beta}{2}%
}\right]  ^{-1}\nonumber\\
&  =\sigma^{\frac{1-\beta}{2}}\left[  \mathcal{N}^{\dag}\left(  \mathcal{N}%
(\sigma)^{-\frac{\beta-1}{2}}\mathcal{N}(\rho)^{-\left(  1-\beta\right)
}\mathcal{N}(\sigma)^{-\frac{\beta-1}{2}}\right)  \right]  ^{-1}\sigma
^{\frac{1-\beta}{2}}\\
&  \leq\sigma^{\frac{1-\beta}{2}}\mathcal{N}^{\dag}\left(  \left[
\mathcal{N}(\sigma)^{-\frac{\beta-1}{2}}\mathcal{N}(\rho)^{-\left(
1-\beta\right)  }\mathcal{N}(\sigma)^{-\frac{\beta-1}{2}}\right]
^{-1}\right)  \sigma^{\frac{1-\beta}{2}}\\
&  =\sigma^{\frac{1-\beta}{2}}\mathcal{N}^{\dag}\left(  \mathcal{N}%
(\sigma)^{\frac{\beta-1}{2}}\mathcal{N}(\rho)^{1-\beta}\mathcal{N}%
(\sigma)^{\frac{\beta-1}{2}}\right)  \sigma^{\frac{1-\beta}{2}},
\end{align}
where the equalities follow from the assumptions that $\rho\in\mathcal{S}%
(\mathcal{H})_{++}$, $\sigma\in\mathcal{B}(\mathcal{H})_{++}$, $\mathcal{N}%
(\rho)\in\mathcal{S}(\mathcal{K})_{++}$, and $\mathcal{N}(\sigma
)\in\mathcal{B}(\mathcal{K})_{++}$, and the inequality is an application of
Choi's inequality.\ By applying \eqref{eq:trace-monotone-1} and the result
above, we find that%
\begin{multline}
\operatorname{Tr}\left\{  \left[  \left[  \sigma^{-\frac{1-\beta}{2}%
}\mathcal{N}^{\dag}\left(  \mathcal{N}(\sigma)^{-\frac{\beta-1}{2}}%
\mathcal{N}(\rho)^{-\left(  1-\beta\right)  }\mathcal{N}(\sigma)^{-\frac
{\beta-1}{2}}\right)  \sigma^{-\frac{1-\beta}{2}}\right]  ^{-1}\right]
^{\frac{1}{1-\beta}}\right\} \\
\leq\operatorname{Tr}\left\{  \left[  \sigma^{\frac{1-\beta}{2}}%
\mathcal{N}^{\dag}\left(  \mathcal{N}(\sigma)^{\frac{\beta-1}{2}}%
\mathcal{N}(\rho)^{1-\beta}\mathcal{N}(\sigma)^{\frac{\beta-1}{2}}\right)
\sigma^{\frac{1-\beta}{2}}\right]  ^{\frac{1}{1-\beta}}\right\}  \leq1,
\end{multline}
where the last inequality is a consequence of what we have previously shown.
The other inequality (\ref{eq:other-trace-ineq-petz}) for $\alpha\in\left(
1,\infty\right)  $ follows from (\ref{eq:trace-ineq-rel-ent-diff}) for
$\alpha\in\left(  1,2\right)  $ by making the same identification as in
(\ref{eq:gamma-subs}): note here that $\gamma\in(1,2)$ if $\alpha\in
(1,\infty)$.
\end{proof}

\bigskip As a corollary, we establish the following trace inequality, which
was left as an open question in \cite{R02}:

\begin{corollary}
For $\rho\in\mathcal{S}(\mathcal{H})_{++}$, $\sigma\in\mathcal{B}%
(\mathcal{H})_{++}$, and a strict CPTP\ map $\mathcal{N}$ acting on
$\mathcal{S}(\mathcal{H})$, the following inequality holds:%
\begin{equation}
\operatorname{Tr}\left\{  \exp\left\{  \log\sigma+\mathcal{N}^{\dag}\left(
\log\mathcal{N}( \rho) -\log\mathcal{N}( \sigma) \right)  \right\}  \right\}
\leq1.
\end{equation}

\end{corollary}

\begin{proof}
This follows by taking the limit $\alpha\nearrow1$ in
Lemma~\ref{lem:trace-ineq-gen} and using \cite[Lemma 24]{SBW14}:%
\begin{multline}
\exp\left\{  \log\sigma+\mathcal{N}^{\dag}\left(  \log\mathcal{N}(\rho
)-\log\mathcal{N}(\sigma)\right)  \right\} \\
=\lim_{\alpha\nearrow1}\left[  \sigma^{\frac{1-\alpha}{2}}\mathcal{N}^{\dag
}\left(  \mathcal{N}(\sigma)^{\frac{\alpha-1}{2}}\mathcal{N}(\rho)^{1-\alpha
}\mathcal{N}(\sigma)^{\frac{\alpha-1}{2}}\right)  \sigma^{\frac{1-\alpha}{2}%
}\right]  ^{\frac{1}{1-\alpha}}.
\end{multline}
The inequality in Lemma~\ref{lem:trace-ineq-gen} is preserved in the limit.
\end{proof}

We can also solve an open question from \cite{Z14b} for some values of
$\alpha$, simply by applying Lemma~\ref{lem:trace-ineq-gen} for the choices
$\rho=\rho_{ABC}$, $\sigma=\rho_{AC}\otimes I_{B}$, and $\mathcal{N}%
=\operatorname{Tr}_{A}$:

\begin{corollary}
\label{cor:trace-ineq}Let $\rho_{ABC}\in\mathcal{S}\left(  \mathcal{H}%
_{ABC}\right)  $. If $\alpha\in\left(  0,1\right)  $, then%
\begin{equation}
\operatorname{Tr}\left\{  \left(  \rho_{AC}^{\frac{1-\alpha}{2}}\rho
_{C}^{\frac{\alpha-1}{2}}\rho_{BC}^{1-\alpha}\rho_{C}^{\frac{\alpha-1}{2}}%
\rho_{AC}^{\frac{1-\alpha}{2}}\right)  ^{\frac{1}{1-\alpha}}\right\}  \leq1.
\label{eq:1st-trace}%
\end{equation}
If $\alpha\in\left(  1/2,1\right)  $, then%
\begin{equation}
\operatorname{Tr}\left\{  \left(  \rho_{AC}^{\frac{1-\alpha}{2\alpha}}\rho
_{C}^{\frac{\alpha-1}{2\alpha}}\rho_{BC}^{\frac{1-\alpha}{\alpha}}\rho
_{C}^{\frac{\alpha-1}{2\alpha}}\rho_{AC}^{\frac{1-\alpha}{2\alpha}}\right)
^{\frac{\alpha}{1-\alpha}}\right\}  \leq1. \label{eq:2nd-trace}%
\end{equation}
If $\rho_{ABC}\in\mathcal{S}\left(  \mathcal{H}_{ABC}\right)  _{++}$, then
\eqref{eq:1st-trace} holds for $\alpha\in\left(  1,2\right)  $ and
\eqref{eq:2nd-trace} holds for $\alpha\in\left(  1,\infty\right)  $.
\end{corollary}

\begin{remark}
Let $\rho_{ABC}\in\mathcal{S}\left(  \mathcal{H}_{ABC}\right)  _{++}$. The
well-known inequality \cite{LR73}%
\begin{equation}
\operatorname{Tr}\left\{  \exp\left\{  \log\rho_{AC}+\log\rho_{BC}-\log
\rho_{C}\right\}  \right\}  \leq1
\end{equation}
follows from Corollary~\ref{cor:trace-ineq} by taking the limit $\alpha
\rightarrow1$ and using the generalized Lie-Trotter product formula
\cite{S85}:%
\begin{equation}
\exp\left\{  \log\rho_{AC}+\log\rho_{BC}-\log\rho_{C}\right\}  =\lim
_{\alpha\rightarrow1}\left(  \rho_{AC}^{\frac{1-\alpha}{2}}\rho_{C}%
^{\frac{\alpha-1}{2}}\rho_{BC}^{1-\alpha}\rho_{C}^{\frac{\alpha-1}{2}}%
\rho_{AC}^{\frac{1-\alpha}{2}}\right)  ^{\frac{1}{1-\alpha}}.
\end{equation}

\end{remark}

The following proposition establishes some important properties of the
$\Delta_{\alpha}(\rho,\sigma,\mathcal{N})$ and $\widetilde{\Delta}_{\alpha
}(\rho,\sigma,\mathcal{N})$ quantities, which were left as open questions in
\cite{SBW14}:

\begin{proposition}
\label{lem:delta-positive}Let $\rho\in\mathcal{S}(\mathcal{H})$, $\sigma
\in\mathcal{B}(\mathcal{H})_{+}$, and let $\mathcal{N}$ be a CPTP\ map. For
all $\alpha\in\left(  0,1\right)  $%
\begin{equation}
\Delta_{\alpha}(\rho,\sigma,\mathcal{N})\geq0, \label{delta-non-neg}%
\end{equation}
with equality occurring if and only if%
\begin{equation}
\rho=\left[  \sigma^{\frac{1-\alpha}{2}}\mathcal{N}^{\dag}\left(
\mathcal{N}(\sigma)^{\frac{\alpha-1}{2}}\mathcal{N}(\rho)^{1-\alpha
}\mathcal{N}(\sigma)^{\frac{\alpha-1}{2}}\right)  \sigma^{\frac{1-\alpha}{2}%
}\right]  ^{\frac{1}{1-\alpha}}. \label{eq:Delta=0}%
\end{equation}
For all $\alpha\in\left(  1/2,1\right)  $%
\begin{equation}
\widetilde{\Delta}_{\alpha}(\rho,\sigma,\mathcal{N})\geq0,
\label{tilde-delta-non-neg}%
\end{equation}
with equality occurring if and only if%
\begin{equation}
\rho=\left[  \sigma^{\frac{1-\alpha}{2\alpha}}\mathcal{N}^{\dag}\left(
\mathcal{N}(\sigma)^{\frac{\alpha-1}{2\alpha}}\mathcal{N}(\rho)^{\frac
{1-\alpha}{\alpha}}\mathcal{N}(\sigma)^{\frac{\alpha-1}{2\alpha}}\right)
\sigma^{\frac{1-\alpha}{2\alpha}}\right]  ^{\frac{\alpha}{1-\alpha}}.
\label{eq:tilde-Delta=0}%
\end{equation}
Let $\rho\in\mathcal{S}(\mathcal{H})_{++}$, $\sigma\in\mathcal{B}%
(\mathcal{H})_{++}$, and let $\mathcal{N}:\mathcal{B}(\mathcal{H}%
)\mapsto\mathcal{B}(\mathcal{K})$ be a CPTP\ map such that $\mathcal{N}%
(\rho)\in\mathcal{S}(\mathcal{K})_{++}$, $\mathcal{N}(\sigma)\in
\mathcal{B}(\mathcal{K})_{++}$. For these choices, the inequality in
\eqref{delta-non-neg} and the equality condition in \eqref{eq:Delta=0}\ hold
if $\alpha\in(1,2)$, and the inequality in \eqref{tilde-delta-non-neg} and
equality condition in \eqref{eq:tilde-Delta=0}\ hold if $\alpha\in(1,\infty)$.
\end{proposition}

\begin{proof}
Note that the quantities $\Delta_{\alpha}$ and $\widetilde{\Delta}_{\alpha}$
defined in (\ref{eq:Renyi-rel-ent-diff}) and (\ref{eq:Renyi-rel-ent-diff-sand}%
), respectively, can be expressed in terms of the $\alpha$-R\'{e}nyi relative
entropy and the $\alpha$-sandwiched R\'{e}nyi relative entropy as follows
\cite{SBW14}:%
\begin{align}
\Delta_{\alpha}(\rho,\sigma,\mathcal{N})  &  =D_{\alpha}\left(  \rho
\middle\Vert\left[  \sigma^{\frac{1-\alpha}{2}}\mathcal{N}^{\dag}\left(
\mathcal{N}(\sigma)^{\frac{\alpha-1}{2}}\mathcal{N}\left(  \rho\right)
^{1-\alpha}\mathcal{N}(\sigma)^{\frac{\alpha-1}{2}}\right)  \sigma
^{\frac{1-\alpha}{2}}\right]  ^{\frac{1}{1-\alpha}}\right)
,\label{eq:ren-rel-ent-diff-1}\\
\widetilde{\Delta}_{\alpha}(\rho,\sigma,\mathcal{N})  &  =\widetilde
{D}_{\alpha}\left(  \rho\middle\Vert\left[  \sigma^{\frac{1-\alpha}{2\alpha}%
}\mathcal{N}^{\dag}\left(  \mathcal{N}(\sigma)^{\frac{\alpha-1}{2\alpha}%
}\mathcal{N(}\rho)^{\frac{1-\alpha}{\alpha}}\mathcal{N(}\sigma)^{\frac
{\alpha-1}{2\alpha}}\right)  \sigma^{\frac{1-\alpha}{2\alpha}}\right]
^{\frac{\alpha}{1-\alpha}}\right)  . \label{eq:ren-rel-ent-diff-2}%
\end{align}
The non-negativity and equality conditions of $\Delta_{\alpha}(\rho
,\sigma,\mathcal{N})$ for $\alpha\in(0,1)\cup(1,2)$, and that of
$\widetilde{\Delta}_{\alpha}(\rho,\sigma,\mathcal{N})$ for $\alpha
\in(1/2,1)\cup(1,\infty)$, then follow by applying Lemmas~\ref{lem:min-max-0}
and \ref{lem:trace-ineq-gen}.
\end{proof}

Appendix~\ref{app:alt-proofs} gives alternate proofs of the inequality in
\eqref{delta-non-neg} for $\alpha\in(1,2)$ and the inequality in
\eqref{tilde-delta-non-neg} for $\alpha\in(1,\infty)$. These were the original
proofs appearing in a preprint version of this paper and might be of
independent interest. Appendix~\ref{app:alt-proofs} gives some other equality conditions.

\section{Sufficiency of quantum channels and R\'enyi generalizations of
relative entropy differences}

\begin{theorem}
\label{thm:channel-suff-result}Let $\rho\in\mathcal{S}(\mathcal{H})_{++}$,
$\sigma\in\mathcal{B}(\mathcal{H})_{++}$, and let $\mathcal{N}$ be a strict
CPTP\ map. Then $\Delta_{\alpha}(\rho,\sigma,\mathcal{N})=0$ and
$\widetilde{\Delta}_{\alpha}(\rho,\sigma,\mathcal{N})=0$ for all $\alpha
\in\left(  0,1\right)  \cup\left(  1,\infty\right)  $ if $\mathcal{N}$ is
sufficient for $\rho$ and $\sigma$. Furthermore, $\mathcal{N}$ is sufficient
for $\rho$ and $\sigma$ if $\Delta_{\alpha}(\rho,\sigma,\mathcal{N})=0$ for
some $\alpha\in\left(  0,1\right)  \cup\left(  1,2\right)  $ or if
$\widetilde{\Delta}_{\alpha}(\rho,\sigma,\mathcal{N})$ for some $\alpha
\in\left(  1/2,1\right)  \cup\left(  1,\infty\right)  $.
\end{theorem}

\begin{proof}
We begin by proving that $\Delta_{\alpha}(\rho,\sigma,\mathcal{N})=0$ for all
$\alpha\in\left(  0,1\right)  \cup\left(  1,\infty\right)  $ if $\mathcal{N}$
is sufficient for $\rho$ and $\sigma$. So suppose that $\mathcal{N}$ is
sufficient for $\rho$ and $\sigma$. Then according to Theorem$~$%
\ref{thm:mos-petz-decomp}, this implies that the decompositions in
(\ref{eq:sufficiency-states-decompose}%
)-(\ref{eq:sufficiency-channel-decompose}) hold. To simplify things, we
exploit the direct sum structure in (\ref{eq:sufficiency-states-decompose}%
)-(\ref{eq:sufficiency-channel-decompose}) and first evaluate the contribution
arising from the $j^{\text{th}}$ block, for any $j$. Then for a given block,
evaluating the formula%
\begin{equation}
\rho^{\alpha}\sigma^{\frac{1-\alpha}{2}}\mathcal{N}^{\dag}\left(
\mathcal{N}(\sigma)^{\frac{\alpha-1}{2}}\mathcal{N}(\rho)^{1-\alpha
}\mathcal{N}(\sigma)^{\frac{\alpha-1}{2}}\right)  \sigma^{\frac{1-\alpha}{2}}%
\end{equation}
gives the following (where we suppress the index $j$, for simplicity)%
\begin{multline}
\left(  p\rho_{L}\otimes\tau_{R}\right)  ^{\alpha}\left(  q\sigma_{L}%
\otimes\tau_{R}\right)  ^{\frac{1-\alpha}{2}}\times\nonumber\\
\left[  \mathcal{U}^{\dag}\otimes\left(  \mathcal{N}^{R}\right)  ^{\dag
}\right]  \left(  \left(  q\mathcal{U}(\sigma_{L})\otimes\mathcal{N}^{R}%
(\tau_{R})\right)  ^{\frac{\alpha-1}{2}}\left(  p\mathcal{U}(\rho_{L}%
)\otimes\mathcal{N}^{R}(\tau_{R})\right)  ^{1-\alpha}\left(  q\mathcal{U}%
(\sigma_{L})\otimes\mathcal{N}^{R}(\tau_{R})\right)  ^{\frac{\alpha-1}{2}%
}\right)  \times\\
\left(  q\sigma_{L}\otimes\tau_{R}\right)  ^{\frac{1-\alpha}{2}}%
\end{multline}%
\begin{multline}
=p\left(  \rho_{L}\otimes\tau_{R}\right)  ^{\alpha}\left(  \sigma_{L}%
\otimes\tau_{R}\right)  ^{\frac{1-\alpha}{2}}\times\\
\left[  \mathcal{U}^{\dag}\otimes\left(  \mathcal{N}^{R}\right)  ^{\dag
}\right]  \left(  \left(  \mathcal{U}(\sigma_{L})\otimes\mathcal{N}^{R}%
(\tau_{R})\right)  ^{\frac{\alpha-1}{2}}\left(  \mathcal{U}(\rho_{L}%
)\otimes\mathcal{N}^{R}(\tau_{R})\right)  ^{1-\alpha}\left(  \mathcal{U}%
(\sigma_{L})\otimes\mathcal{N}^{R}(\tau_{R})\right)  ^{\frac{\alpha-1}{2}%
}\right)  \times\\
\left(  \sigma_{L}\otimes\tau_{R}\right)  ^{\frac{1-\alpha}{2}}%
\end{multline}%
\begin{multline}
=p\left(  (\rho_{L})^{\alpha}(\sigma_{L})^{\frac{1-\alpha}{2}}\otimes\tau
_{R}^{\left(  \alpha+1\right)  /2}\right)  \times\\
\left[  \mathcal{U}^{\dag}\otimes\left(  \mathcal{N}^{R}\right)  ^{\dag
}\right]  \left(  \left(  \mathcal{U}(\sigma_{L})^{\frac{\alpha-1}{2}%
}\mathcal{U}(\rho_{L})^{1-\alpha}\mathcal{U}(\sigma_{L})^{\frac{\alpha-1}{2}%
}\otimes\mathcal{N}^{R}(\tau_{R})^{\frac{\alpha-1}{2}}\mathcal{N}^{R}(\tau
_{R})^{1-\alpha}\mathcal{N}^{R}(\tau_{R})^{\frac{\alpha-1}{2}}\right)
\right)  \times\\
\left(  \sigma_{L}\otimes\tau_{R}\right)  ^{\frac{1-\alpha}{2}}.
\end{multline}
Continuing, the last line is equal to%
\begin{align}
&  p\left(  (\rho_{L})^{\alpha}(\sigma_{L})^{\frac{1-\alpha}{2}}\otimes
\tau_{R}^{\left(  \alpha+1\right)  /2}\right)  \left[  \mathcal{U}^{\dag
}\otimes\left(  \mathcal{N}^{R}\right)  ^{\dag}\right]  \left(  \mathcal{U}%
\left[  (\sigma_{L})^{\frac{\alpha-1}{2}}(\rho_{L})^{1-\alpha}(\sigma
_{L})^{\frac{\alpha-1}{2}}\right]  \otimes I_{R}\right)  \left(  \sigma
_{L}\otimes\tau_{R}\right)  ^{\frac{1-\alpha}{2}}\nonumber\\
&  =p\left(  (\rho_{L})^{\alpha}(\sigma_{L})^{\frac{1-\alpha}{2}}\otimes
\tau_{R}^{\left(  \alpha+1\right)  /2}\right)  \left(  \left[  (\sigma
_{L})^{\frac{\alpha-1}{2}}(\rho_{L})^{1-\alpha}(\sigma_{L})^{\frac{\alpha
-1}{2}}\right]  \otimes I_{R}\right)  \left(  \sigma_{L}\otimes\tau
_{R}\right)  ^{\frac{1-\alpha}{2}}\nonumber\\
&  =p\left(  (\rho_{L})^{\alpha}(\sigma_{L})^{\frac{1-\alpha}{2}}(\sigma
_{L})^{\frac{\alpha-1}{2}}(\rho_{L})^{1-\alpha}(\sigma_{L})^{\frac{\alpha
-1}{2}}(\sigma_{L})^{\frac{1-\alpha}{2}}\otimes\tau_{R}^{\left(
\alpha+1\right)  /2}\tau_{R}^{\frac{1-\alpha}{2}}\right) \nonumber\\
&  =p\left(  (\rho_{L})^{\alpha}(\rho_{L})^{1-\alpha}\otimes\tau_{R}\right)
\nonumber\\
&  =p\rho_{L}\otimes\tau_{R}.
\end{align}
Taking the trace gives $p$ and thus we find that each block has trace equal to
$p\left(  j\right)  $, so that%
\begin{equation}
\Delta_{\alpha}(\rho,\sigma,\mathcal{N})=\frac{1}{\alpha-1}\log\sum_{j}p(j)=0.
\end{equation}
So this proves that $\Delta_{\alpha}(\rho,\sigma,\mathcal{N})=0$ for all
$\alpha\in\left(  0,1\right)  \cup\left(  1,\infty\right)  $ if $\mathcal{N}$
is sufficient for $\rho$ and $\sigma$.

We now prove that $\widetilde{\Delta}_{\alpha}(\rho,\sigma,\mathcal{N})=0$ for
all $\alpha\in\left(  0,1\right)  \cup\left(  1,\infty\right)  $ if
$\mathcal{N}$ is sufficient for $\rho$ and $\sigma$. As above, the
decompositions in (\ref{eq:sufficiency-states-decompose}%
)-(\ref{eq:sufficiency-channel-decompose}) hold and we exploit this direct sum
structure again. As in the previous proof, it suffices to evaluate the
contribution arising from each block $j$, for any $j$, in the direct-sum
decomposition. Evaluating the formula%
\begin{equation}
\rho^{\frac{1}{2}}\sigma^{\frac{1-\alpha}{2\alpha}}\mathcal{N}^{\dag}\left(
\mathcal{N}(\sigma)^{\frac{\alpha-1}{2\alpha}}\mathcal{N}(\rho)^{1-\alpha
}\mathcal{N}(\sigma)^{\frac{\alpha-1}{2\alpha}}\right)  \sigma^{\frac
{1-\alpha}{2\alpha}}\rho^{\frac{1}{2}},
\end{equation}
for the $j$th block, (where we once again suppress the index $j$) gives%
\begin{multline}
\left(  p\rho_{L}\otimes\tau_{R}\right)  ^{\frac{1}{2}}\left(  q\sigma
_{L}\otimes\tau_{R}\right)  ^{\frac{1-\alpha}{2\alpha}}\times\\
\left[  \mathcal{U}^{\dag}\otimes\left(  \mathcal{N}^{R}\right)  ^{\dag
}\right]  \left(  \left(  q\mathcal{U}(\sigma_{L})\otimes\mathcal{N}^{R}%
(\tau_{R})\right)  ^{\frac{\alpha-1}{2\alpha}}\left(  p\mathcal{U}(\rho
_{L})\otimes\mathcal{N}^{R}(\tau_{R})\right)  ^{\frac{1-\alpha}{\alpha}%
}\left(  q\mathcal{U}(\sigma_{L})\otimes\mathcal{N}^{R}(\tau_{R})\right)
^{\frac{\alpha-1}{2\alpha}}\right)  \times\\
\left(  q\sigma_{L}\otimes\tau_{R}\right)  ^{\frac{1-\alpha}{2\alpha}}\left(
p\rho_{L}\otimes\tau_{R}\right)  ^{\frac{1}{2}}. \label{eq:jth-block-starting}%
\end{multline}
Then proceeding very similarly as before, we conclude that
(\ref{eq:jth-block-starting}) is equal to%
\begin{equation}
p^{\frac{1}{\alpha}}\rho_{L}^{\frac{1}{\alpha}}\otimes\tau_{R}^{\frac
{1}{\alpha}}.
\end{equation}
We evaluate evaluate the norm $\left\Vert \cdot\right\Vert _{\alpha}^{\alpha}%
$\ of the $j$th block to be%
\begin{equation}
\left\Vert p^{\frac{1}{\alpha}}\rho_{L}^{\frac{1}{\alpha}}\otimes\tau
_{R}^{\frac{1}{\alpha}}\right\Vert _{\alpha}^{\alpha}=p.
\end{equation}
As a consequence, we find that%
\begin{equation}
\widetilde{\Delta}_{\alpha}(\rho,\sigma,\mathcal{N})=\frac{1}{\alpha-1}%
\log\left(  \sum_{j}p(j)\right)  =0,
\end{equation}
where we used that $\left\Vert A\oplus B\right\Vert _{\alpha}^{\alpha
}=\left\Vert A\right\Vert _{\alpha}^{\alpha}+\left\Vert B\right\Vert _{\alpha
}^{\alpha}$.

Now suppose that $\Delta_{\alpha}(\rho,\sigma,\mathcal{N})=0$ for some
$\alpha\in\left(  0,1\right)  \cup(1,2)$. From
Proposition~\ref{lem:delta-positive}, we have that%
\begin{equation}
\rho=\left[  \sigma^{\frac{1-\alpha}{2}}\mathcal{N}^{\dag}\left(
\mathcal{N}(\sigma)^{\frac{\alpha-1}{2}}\mathcal{N}(\rho)^{1-\alpha
}\mathcal{N}(\sigma)^{\frac{\alpha-1}{2}}\right)  \sigma^{\frac{1-\alpha}{2}%
}\right]  ^{\frac{1}{1-\alpha}}, \label{eq:step-ren-rel-ent-1}%
\end{equation}
which is equivalent to%
\begin{equation}
\sigma^{\frac{\alpha-1}{2}}\rho^{1-\alpha}\sigma^{\frac{\alpha-1}{2}%
}=\mathcal{N}^{\dag}\left(  \mathcal{N}(\sigma)^{\frac{\alpha-1}{2}%
}\mathcal{N}(\rho)^{1-\alpha}\mathcal{N}(\sigma)^{\frac{\alpha-1}{2}}\right)
.
\end{equation}
Multiply both sides by $\sigma$ and take the trace to get%
\begin{align}
\text{Tr}\left\{  \sigma\sigma^{\frac{\alpha-1}{2}}\rho^{1-\alpha}%
\sigma^{\frac{\alpha-1}{2}}\right\}   &  =\text{Tr}\left\{  \sigma
\mathcal{N}^{\dag}\left(  \mathcal{N}(\sigma)^{\frac{\alpha-1}{2}}%
\mathcal{N}(\rho)^{1-\alpha}\mathcal{N}(\sigma)^{\frac{\alpha-1}{2}}\right)
\right\} \nonumber\\
\Leftrightarrow\text{Tr}\left\{  \rho^{1-\alpha}\sigma^{\alpha}\right\}   &
=\text{Tr}\left\{  \mathcal{N}(\sigma)\mathcal{N}(\sigma)^{\frac{\alpha-1}{2}%
}\mathcal{N}(\rho)^{1-\alpha}\mathcal{N}(\sigma)^{\frac{\alpha-1}{2}}\right\}
\nonumber\\
\Leftrightarrow\text{Tr}\left\{  \rho^{1-\alpha}\sigma^{\alpha}\right\}   &
=\text{Tr}\left\{  \mathcal{N}(\rho)^{1-\alpha}\mathcal{N}(\sigma)^{\alpha
}\right\} \nonumber\\
\Leftrightarrow D_{\alpha}(\sigma\Vert\rho)  &  =D_{\alpha}(\mathcal{N}%
(\sigma)\Vert\mathcal{N}(\rho)). \label{eq:step-ren-rel-ent-last}%
\end{align}
The last line is an equality of $f$-divergences (see
Remark~\ref{rem:renyi-f-div}), which by Lemma~\ref{lem-equiv} allows us to
conclude that $\mathcal{N}$ is sufficient for $\rho$ and $\sigma$.

Suppose that $\widetilde{\Delta}_{\alpha}(\rho,\sigma,\mathcal{N})=0$ for some
$\alpha\in\left(  1/2,1\right)  \cup(1,\infty)$. From
Proposition~\ref{lem:delta-positive}, we have that%
\begin{equation}
\rho=\left[  \sigma^{\frac{1-\alpha}{2\alpha}}\mathcal{N}^{\dag}\left(
\mathcal{N}(\sigma)^{\frac{\alpha-1}{2\alpha}}\mathcal{N}(\rho)^{\frac
{1-\alpha}{\alpha}}\mathcal{N}(\sigma)^{\frac{\alpha-1}{2\alpha}}\right)
\sigma^{\frac{1-\alpha}{2\alpha}}\right]  ^{\frac{\alpha}{1-\alpha}},
\end{equation}
which can be rewritten as%
\begin{equation}
\rho=\left[  \sigma^{\frac{1-\gamma}{2}}\mathcal{N}^{\dag}\left(
\mathcal{N}(\sigma)^{\frac{\gamma-1}{2}}\mathcal{N}(\rho)^{1-\gamma
}\mathcal{N}(\sigma)^{\frac{\gamma-1}{2}}\right)  \sigma^{\frac{1-\gamma}{2}%
}\right]  ^{\frac{1}{1-\gamma}},
\end{equation}
with $\gamma$ defined in (\ref{eq:gamma-subs}), so that $\gamma\in
(0,1)\cup(1,2)$. By the development in
\eqref{eq:step-ren-rel-ent-1}-\eqref{eq:step-ren-rel-ent-last}, we can
conclude that $\mathcal{N}$ is sufficient for $\rho$ and $\sigma$.
\end{proof}

\bigskip

The following statement is a direct corollary of the correspondence in
\eqref{eq:rel-ent-cmi}-\eqref{eq:rel-ent-cmi-1} and
Theorem~\ref{thm:channel-suff-result}:

\begin{corollary}
\label{thm:Renyi-cmi-markov}Let $\rho_{ABC}\in\mathcal{S}\left(
\mathcal{H}_{ABC}\right)  _{++}$. Then $I_{\alpha}(A;B|C)_{\rho}=0$ and
$\widetilde{I}_{\alpha}(A;B|C)_{\rho}=0$ for all $\alpha\in(0,1)\cup
(1,\infty)$ if $\rho_{ABC}$ is a short quantum Markov chain $A-C-B$.
Furthermore, $\rho_{ABC}$ is a short quantum Markov chain $A-C-B\mathcal{\ }%
$if $I_{\alpha}(A;B|C)_{\rho}=0$ for some $\alpha\in\left(  0,1\right)
\cup\left(  1,2\right)  $ or if $\widetilde{I}_{\alpha}(A;B|C)_{\rho}=0$ for
some $\alpha\in\left(  1/2,1\right)  \cup\left(  1,\infty\right)  $.
\end{corollary}

\section{Quantum Markov chains, sufficiency of quantum channels, and min- and
max-information measures}

\label{sec:min-max}

\subsection{Min- and max- generalizations of a relative entropy difference}

We consider the following generalizations of a relative entropy difference,
motivated by the developments in \cite{BSW14,SBW14}.

\begin{definition}
\label{def:CMI-quantities}Let $\rho\in\mathcal{S}(\mathcal{H})_{++}$,
$\sigma\in\mathcal{B}(\mathcal{H})_{++}$, and let $\mathcal{N}$ be a strict
CPTP map. Then,
\begin{equation}
\Delta_{\min}(\rho,\sigma,\mathcal{N})\equiv D_{\min}(\rho\Vert\mathcal{R}%
_{\sigma,\mathcal{N}}(\mathcal{N(}\rho))),
\end{equation}
where $\mathcal{R}_{\sigma,\mathcal{N}}$ is the Petz recovery channel defined
in \eqref{eq:Petz-map} and
\begin{equation}
\Delta_{\max}(\rho,\sigma,\mathcal{N})\equiv D_{\max}(\rho\Vert\mathcal{R}%
_{\sigma,\mathcal{N}}(\mathcal{N(}\rho))).
\end{equation}

\end{definition}

\begin{theorem}
\label{thm:renyi-rel-ent-min-max-equality}For $\rho\in\mathcal{S}%
(\mathcal{H})_{++}$, $\sigma\in\mathcal{B}(\mathcal{H})_{++}$, and a strict
CPTP map $\mathcal{N}$,
\begin{equation}
\Delta_{\min}(\rho,\sigma,\mathcal{N})\geq0,\ \ \ \ \quad\Delta_{\max}%
(\rho,\sigma,\mathcal{N})\geq0,
\end{equation}
with equality holding if and only if $\mathcal{N}$ is sufficient for $\rho$
and $\sigma$.
\end{theorem}

\begin{proof}
The non-negativity conditions follow from the fact that $\mathcal{R}%
_{\sigma,\mathcal{N}}(\mathcal{N(}\rho))$ is a density operator (since
$\mathcal{R}_{\sigma,\mathcal{N}}$ is a CPTP map) and
Lemma~\ref{lem:min-max-0}. The equality conditions also follow from
Lemma~\ref{lem:min-max-0} and the fact that $\mathcal{N}$ is sufficient for
$\rho$ and $\sigma$ if and only if $\mathcal{R}_{\sigma,\mathcal{N}%
}(\mathcal{N(}\rho))=\rho$.
\end{proof}

\subsection{Min- and max- generalizations of conditional mutual information}

\begin{definition}
[\cite{BSW14}]For a tripartite state $\rho_{ABC}\in\mathcal{S}\left(
\mathcal{H}_{ABC}\right)  _{++}$ the max-conditional mutual information is
defined as follows:%
\begin{align}
I_{\max}(A;B|C)_{\rho} &  \equiv D_{\max}\left(  \rho_{ABC}\middle\Vert
\rho_{AC}^{\frac{1}{2}}\rho_{C}^{-\frac{1}{2}}\rho_{BC}\rho_{C}^{-\frac{1}{2}%
}\rho_{AC}^{\frac{1}{2}}\right)  \label{imax1}\\
&  =\inf\left\{  \lambda:\,\rho_{ABC}\leq\exp(\lambda)\rho_{AC}^{\frac{1}{2}%
}\rho_{C}^{-\frac{1}{2}}\rho_{BC}\rho_{C}^{-\frac{1}{2}}\rho_{AC}^{\frac{1}%
{2}}\right\}  \label{imax2}\\
&  =2\log\left\Vert \rho_{ABC}^{\frac{1}{2}}\rho_{AC}^{-\frac{1}{2}}\rho
_{C}^{\frac{1}{2}}\rho_{BC}^{-\frac{1}{2}}\right\Vert _{\infty},\label{imax3}%
\end{align}
and the min-conditional mutual information is defined as follows:%
\begin{align}
I_{\min}(A;B|C)_{\rho} &  \equiv D_{\min}\left(  \rho_{ABC}\middle\Vert
\rho_{AC}^{\frac{1}{2}}\rho_{C}^{-\frac{1}{2}}\rho_{BC}\rho_{C}^{-\frac{1}{2}%
}\rho_{AC}^{\frac{1}{2}}\right)  \label{imin1}\\
&  =-\log F\left(  \rho_{ABC},\rho_{AC}^{\frac{1}{2}}\rho_{C}^{-\frac{1}{2}%
}\rho_{BC}\rho_{C}^{-\frac{1}{2}}\rho_{AC}^{\frac{1}{2}}\right)
\label{imin2}\\
&  =-2\log\left\Vert \rho_{ABC}^{\frac{1}{2}}\rho_{AC}^{\frac{1}{2}}\rho
_{C}^{-\frac{1}{2}}\rho_{BC}^{\frac{1}{2}}\right\Vert _{1}.\label{imin3}%
\end{align}

\end{definition}

As observed in \cite{BSW14}, we have that%
\begin{equation}
I_{\max}(A;B|C)_{\rho},I_{\min}(A;B|C)_{\rho}\geq0, \label{eq:I_min_max_geq0}%
\end{equation}
due to Lemma~\ref{lem:min-max-0} and the fact that Tr$\{\rho_{AC}^{\frac{1}%
{2}}\rho_{C}^{-\frac{1}{2}}\rho_{BC}\rho_{C}^{-\frac{1}{2}}\rho_{AC}^{\frac
{1}{2}}\}=1$. These quantities are special cases of those from
Definition~\ref{def:CMI-quantities}, by using (\ref{eq:rel-ent-cmi-1}). As
such, the following is a corollary of
Theorem~\ref{thm:renyi-rel-ent-min-max-equality}:

\begin{corollary}
\label{thm-imax-imin} Let $\rho_{ABC}\in\mathcal{S}\left(  \mathcal{H}%
_{ABC}\right)  _{++}$. Then each of the following identities hold if and only
if $\rho_{ABC}$ is a short quantum Markov chain $A-C-B$:
\begin{align}
I_{\max}(A;B|C)  &  =0,\label{imax-a}\\
I_{\min}(A;B|C)  &  =0. \label{imin-b}%
\end{align}

\end{corollary}

\section{Conclusion and open questions}

We have shown that the $\alpha$-R\'{e}nyi\ quantities $\Delta_{\alpha}%
(\rho,\sigma,\mathcal{N})$ and $\widetilde{\Delta}_{\alpha}(\rho
,\sigma,\mathcal{N})$ from \cite{SBW14} are non-negative and equal to zero if
and only if $\mathcal{N}$ is sufficient for $\rho$ and $\sigma$. As a
consequence, we find that the $\alpha$-R\'{e}nyi conditional mutual
informations $I_{\alpha}(A;B|C)_{\rho}$ and $\widetilde{I}_{\alpha
}(A;B|C)_{\rho}$ from \cite{BSW14} are equal to zero if and only if
$\rho_{ABC}$ is a short quantum Markov chain $A-C-B$. Moreover, we have solved
some open questions from \cite{R02,Z14b}.

There are some interesting open questions to consider going forward from here.
We would like to know the equality conditions for monotonicity of the
sandwiched R\'{e}nyi relative entropies, i.e., for which triples $(\rho
,\sigma,\mathcal{N})$ is it true that%
\begin{equation}
\widetilde{D}_{\alpha}(\rho\Vert\sigma)=\widetilde{D}_{\alpha}\left(
\mathcal{N}(\rho)\Vert\mathcal{N}(\sigma)\right)  \ ?
\end{equation}
Apparently we cannot solve this question using the methods of \cite{HMPB11}
because we cannot represent the sandwiched R\'{e}nyi relative entropy as an
$f$-divergence. Presumably, this equality occurs if and only if $\mathcal{N}$
is sufficient for $\rho$ and $\sigma$, but it remains to be proved. Next, is
there a characterization of states for which $I_{\alpha}(A;B|C)_{\rho}$ and
$\widetilde{I}_{\alpha}(A;B|C)_{\rho}$ are nearly equal to zero? Also, is
there a characterization of triples $(\rho,\sigma,\mathcal{N})$ for which
$\Delta_{\alpha}(\rho,\sigma,\mathcal{N})$ and $\widetilde{\Delta}_{\alpha
}(\rho,\sigma,\mathcal{N})$ are nearly equal to zero? Presumably the former
has to do with $\rho_{ABC}$ being close to the \textquotedblleft Petz
recovered\textquotedblright\ $\rho_{BC}$ and the latter has to do with $\rho$
being close to the \textquotedblleft Petz recovered\textquotedblright\ $\rho$,
as recent developments in \cite{FR14,BLW14} might suggest.

\bigskip\textbf{Note}: After the posting of a preprint of this work to the
arXiv, there have been improvements of the results detailed here for values of
$\alpha<1$ \cite{W15,DW15,JRSWW15}, which address some of the open questions
mentioned above for these values of $\alpha$. In particular, the results of
\cite{W15} provide stronger lower bounds for $\widetilde{I}_{\alpha
}(A;B|C)_{\rho}$ and $\widetilde{\Delta}_{\alpha}(\rho,\sigma,\mathcal{N})$
when $\alpha\in(1/2,1)$. This was further improved in \cite{JRSWW15}. The
results of \cite{DW15} provide stronger lower bounds for $I_{\alpha
}(A;B|C)_{\rho}$ and $\Delta_{\alpha}(\rho,\sigma,\mathcal{N})$ when
$\alpha\in(0,1)$.

\bigskip\textbf{Acknowledgements.} We acknowledge helpful discussions with
Mario Berta, Kaushik Seshadreesan, and Marco Tomamichel and thank Perla Sousi
for a helpful discussion on Markov chains. We are especially grateful to Milan
Mosonyi for his very careful reading of our paper, pointing out a problem with
our former justification of Lemma~\ref{lem:trace-ineq-gen}\ for $\alpha
\in\left(  1,2\right)  $, and for communicating many other observations about
our paper that go beyond those stated here. We are indebted to an anonymous
referee for suggesting to use Choi's inequality to prove some of the
statements in Lemma~\ref{lem:trace-ineq-gen}. We acknowledge support from the
Peter Whittle Fund, which helped to enable this research. MMW\ is grateful to
the Statistical Laboratory in the Center for Mathematical Sciences at the
University of Cambridge for hosting him for a research visit during January
2015. MMW acknowledges support from startup funds from the Department of
Physics and Astronomy at LSU, the NSF\ under Award No.~CCF-1350397, and the
DARPA Quiness Program through US Army Research Office award W31P4Q-12-1-0019.

\appendix

\section{Alternative Proof of Non-Negativity of $\Delta_{\alpha}$ and
$\widetilde{\Delta}_{\alpha}$ quantities}

\label{app:alt-proofs}Here we provide an alternative proof of the
non-negativity of the $\Delta_{\alpha}$ and $\widetilde{\Delta}_{\alpha}$
quantities, which appeared in the original preprint version of our work:

\begin{proposition}
Let $\rho\in\mathcal{S}(\mathcal{H})_{++}$, $\sigma\in\mathcal{B}%
(\mathcal{H})_{++}$, and let $\mathcal{N}$ be a strict CPTP\ map. Then
\eqref{delta-non-neg} holds for $\alpha\in\left(  1,2\right)  $ and equality
occurs if and only if%
\begin{equation}
\left[  \mathcal{N}(\sigma)^{\frac{\alpha-1}{2}}\mathcal{N}\left(
\sigma^{\frac{1-\alpha}{2}}\rho^{\alpha}\sigma^{\frac{1-\alpha}{2}}\right)
\mathcal{N}(\sigma)^{\frac{\alpha-1}{2}}\right]  ^{\frac{1}{\alpha}%
}=\mathcal{N}(\rho). \label{eq:Delta=0-other-a}%
\end{equation}
Furthermore, for the same choices of $\rho$, $\sigma$, and $\mathcal{N}$,
\eqref{tilde-delta-non-neg} holds for $\alpha\in\left(  1,\infty\right)  $.
\end{proposition}

\begin{proof}
We first prove the non-negativity of $\Delta_{\alpha}(\rho,\sigma
,\mathcal{N})$ for $\alpha\in(1,2)$, $\rho\in\mathcal{S}(\mathcal{H})_{++}$,
$\sigma\in\mathcal{B}(\mathcal{H})_{++}$, and $\mathcal{N}$ a strict
CPTP\ map. Using the definition (\ref{eq:Renyi-rel-ent-diff}) of
$\Delta_{\alpha}(\rho,\sigma,\mathcal{N})$, cyclicity of trace, and the
definition of the adjoint map, we can express $\Delta_{\alpha}(\rho
,\sigma,\mathcal{N})$ as%
\begin{align}
\Delta_{\alpha}(\rho,\sigma,\mathcal{N})  &  =\frac{1}{\alpha-1}\log
\text{Tr}\left\{  \mathcal{N}(\sigma)^{\frac{\alpha-1}{2}}\mathcal{N}\left(
\sigma^{\frac{1-\alpha}{2}}\rho^{\alpha}\sigma^{\frac{1-\alpha}{2}}\right)
\mathcal{N}(\sigma)^{\frac{\alpha-1}{2}}\mathcal{N}(\rho)^{1-\alpha}\right\}
\label{eq:Delta-rewrite}\\
&  =D_{\alpha}\left(  \left[  \mathcal{N}(\sigma)^{\frac{\alpha-1}{2}%
}\mathcal{N}\left(  \sigma^{\frac{1-\alpha}{2}}\rho^{\alpha}\sigma
^{\frac{1-\alpha}{2}}\right)  \mathcal{N}(\sigma)^{\frac{\alpha-1}{2}}\right]
^{\frac{1}{\alpha}}\middle\Vert\mathcal{N}(\rho)\right)
\label{eq:Delta-rewrite-1}%
\end{align}
where%
\begin{equation}
D_{\alpha}(P\Vert Q)\equiv\frac{1}{\alpha-1}\log\text{Tr}\left\{  P^{\alpha
}Q^{1-\alpha}\right\}
\end{equation}
is the R\'{e}nyi relative entropy between two positive definite operators $P$
and $Q$. It is known that $D_{\alpha}(P\Vert Q)\geq0$ if $P$ and $Q$ are such
that%
\begin{equation}
\text{Tr}\left\{  P\right\}  \geq\text{Tr}\left\{  Q\right\}  =1.
\end{equation}
This is because $D_{\alpha}(P\Vert Q)$ is monotone under quantum channels for
$\alpha\in\left(  1,2\right)  $, and one such quantum channel is the trace
operation:%
\begin{align}
D_{\alpha}(P\Vert Q)  &  \geq D_{\alpha}\left(  \text{Tr}\left\{  P\right\}
\Vert\text{Tr}\left\{  Q\right\}  \right) \\
&  =\frac{1}{\alpha-1}\log\left[  \left[  \text{Tr}\left\{  P\right\}
\right]  ^{\alpha}\left[  \text{Tr}\left\{  Q\right\}  \right]  ^{1-\alpha
}\right] \\
&  \geq0. \label{eq:renyi-non-neg}%
\end{align}
If $D_{\alpha}(P\Vert Q)=0$ for some $\alpha\in\left(  1,2\right)  $ and $P$
and $Q$ such that $\operatorname{Tr}\left\{  P\right\}  \geq\operatorname{Tr}%
\left\{  Q\right\}  $, then it is known that $P=Q$ (one can deduce this, e.g.,
from the above and \cite[Theorem~5.1]{HMPB11}). Hence, to prove the
non-negativity of $\Delta_{\alpha}(\rho,\sigma,\mathcal{N})$ for $\alpha
\in\left(  1,2\right)  $, it suffices to prove that%
\begin{equation}
\text{Tr}\left\{  \left[  \mathcal{N}(\sigma)^{\frac{\alpha-1}{2}}%
\mathcal{N}\left(  \sigma^{\frac{1-\alpha}{2}}\rho^{\alpha}\sigma
^{\frac{1-\alpha}{2}}\right)  \mathcal{N}(\sigma)^{\frac{\alpha-1}{2}}\right]
^{\frac{1}{\alpha}}\right\}  \geq\text{Tr}\left\{  \mathcal{N}(\rho)\right\}
=1.
\end{equation}
Theorem~1.1 of \cite{Hiai20131568} establishes that the map%
\begin{equation}
\left(  P,Q\right)  \mapsto\text{Tr}\left\{  \left[  Q^{\frac{\alpha-1}{2}%
}PQ^{\frac{\alpha-1}{2}}\right]  ^{\frac{1}{\alpha}}\right\}
\end{equation}
is jointly concave for positive definite $P$ and $Q$ when $\alpha\in\left(
1,2\right)  $. A straightforward argument allows to conclude its joint
concavity for $\alpha\in\left(  1,2\right)  $ and positive semidefinite $P$
and $Q$. Indeed, let $\varepsilon>0$, $\left\{  P_{x}\right\}  $ and $\left\{
Q_{x}\right\}  $ be sets of positive semidefinite operators, let $p_{X}(x)$ be
a probability distribution, and let $\overline{P}\equiv\sum_{x}p_{X}(x)P_{x}$
and $\overline{Q}\equiv\sum_{x}p_{X}(x)Q_{x}$ . Consider that%
\begin{align}
&  \sum_{x}p_{X}(x)\text{Tr}\left\{  \left[  Q_{x}^{\frac{\alpha-1}{2}}%
P_{x}Q_{x}^{\frac{\alpha-1}{2}}\right]  ^{\frac{1}{\alpha}}\right\}
\nonumber\\
&  \leq\sum_{x}p_{X}(x)\text{Tr}\left\{  \left[  Q_{x}^{\frac{\alpha-1}{2}%
}\left(  P_{x}+\varepsilon I\right)  Q_{x}^{\frac{\alpha-1}{2}}\right]
^{\frac{1}{\alpha}}\right\} \nonumber\\
&  =\sum_{x}p_{X}(x)\text{Tr}\left\{  \left[  \left(  P_{x}+\varepsilon
I\right)  ^{\frac{1}{2}}Q_{x}^{\alpha-1}\left(  P_{x}+\varepsilon I\right)
^{\frac{1}{2}}\right]  ^{\frac{1}{\alpha}}\right\} \nonumber\\
&  \leq\sum_{x}p_{X}(x)\text{Tr}\left\{  \left[  \left(  P_{x}+\varepsilon
I\right)  ^{\frac{1}{2}}\left(  Q_{x}+\varepsilon I\right)  ^{\alpha-1}\left(
P_{x}+\varepsilon I\right)  ^{\frac{1}{2}}\right]  ^{\frac{1}{\alpha}}\right\}
\nonumber\\
&  =\sum_{x}p_{X}(x)\text{Tr}\left\{  \left[  \left(  Q_{x}+\varepsilon
I\right)  ^{\frac{\alpha-1}{2}}\left(  P_{x}+\varepsilon I\right)  \left(
Q_{x}+\varepsilon I\right)  ^{\frac{\alpha-1}{2}}\right]  ^{\frac{1}{\alpha}%
}\right\} \nonumber\\
&  \leq\text{Tr}\left\{  \left[  \left(  \overline{Q}+\varepsilon I\right)
^{\frac{\alpha-1}{2}}\left(  \overline{P}+\varepsilon I\right)  \left(
\overline{Q}+\varepsilon I\right)  ^{\frac{\alpha-1}{2}}\right]  ^{\frac
{1}{\alpha}}\right\}  .
\end{align}
The first inequality follows because $P_{x}\leq P_{x}+\varepsilon I$ and
because $\operatorname{Tr}\left\{  f(A)\right\}  \leq\operatorname{Tr}\left\{
f(B)\right\}  $ for $A\leq B$ and $f(x)=x^{1/\alpha}$ a monotone
non-decreasing function on $[0,\infty)$. The next inequality follows because
$x^{\alpha-1}$ is operator monotone for $\alpha\in\left(  1,2\right)  $ and
for the same reason as above. The final inequality is a consequence of
Theorem~1.1 of \cite{Hiai20131568}. By taking the limit $\varepsilon\searrow
0$, we can conclude that%
\begin{equation}
\sum_{x}p_{X}(x)\text{Tr}\left\{  \left[  Q_{x}^{\frac{\alpha-1}{2}}P_{x}%
Q_{x}^{\frac{\alpha-1}{2}}\right]  ^{\frac{1}{\alpha}}\right\}  \leq
\text{Tr}\left\{  \left[  \overline{Q}^{(\alpha-1)/2}\overline{P}%
\,\overline{Q}^{(\alpha-1)/2}\right]  ^{\frac{1}{\alpha}}\right\}  .
\label{eq:Hiai-concave}%
\end{equation}
Now consider the following chain of equalities:%
\begin{align}
1  &  =\text{Tr}\left\{  \left[  \sigma^{\frac{\alpha-1}{2}}\sigma
^{\frac{1-\alpha}{2}}\rho^{\alpha}\sigma^{\frac{1-\alpha}{2}}\sigma
^{\frac{\alpha-1}{2}}\right]  ^{\frac{1}{\alpha}}\right\} \nonumber\\
&  =\text{Tr}\left\{  \left[  \sigma^{\frac{\alpha-1}{2}}\sigma^{\frac
{1-\alpha}{2}}\rho^{\alpha}\sigma^{\frac{1-\alpha}{2}}\sigma^{\frac{\alpha
-1}{2}}\otimes|0\rangle\langle0|_{E}\right]  ^{\frac{1}{\alpha}}\right\}
\nonumber\\
&  =\text{Tr}\left\{  \left[  \left(  \sigma^{\frac{\alpha-1}{2}}%
\otimes|0\rangle\langle0|_{E}\right)  \left(  \sigma^{\frac{1-\alpha}{2}}%
\rho^{\alpha}\sigma^{\frac{1-\alpha}{2}}\otimes|0\rangle\langle0|_{E}\right)
\left(  \sigma^{\frac{\alpha-1}{2}}\otimes|0\rangle\langle0|_{E}\right)
\right]  ^{\frac{1}{\alpha}}\right\} \nonumber\\
&  =\text{Tr}\left\{  \left[  U\left(  \sigma^{\frac{\alpha-1}{2}}%
\otimes|0\rangle\langle0|_{E}\right)  U^{\dag}U\left(  \sigma^{\frac{1-\alpha
}{2}}\rho^{\alpha}\sigma^{\frac{1-\alpha}{2}}\otimes|0\rangle\langle
0|_{E}\right)  U^{\dag}U\left(  \sigma^{\frac{\alpha-1}{2}}\otimes
|0\rangle\langle0|_{E}\right)  U^{\dag}\right]  ^{\frac{1}{\alpha}}\right\}
\nonumber\\
&  =\text{Tr}\left\{  \left[  \left[  U\left(  \sigma\otimes|0\rangle
\langle0|_{E}\right)  U^{\dag}\right]  ^{\frac{\alpha-1}{2}}U\left(
\sigma^{\frac{1-\alpha}{2}}\rho^{\alpha}\sigma^{\frac{1-\alpha}{2}}%
\otimes|0\rangle\langle0|_{E}\right)  U^{\dag}\left[  U\left(  \sigma
\otimes|0\rangle\langle0|_{E}\right)  U^{\dag}\right]  ^{\frac{\alpha-1}{2}%
}\right]  ^{\frac{1}{\alpha}}\right\} \nonumber\\
&  =\frac{1}{d_{E^{\prime}}^{2}}\sum_{i}\text{Tr}\left\{  \left[  \left[
K^{i}\right]  ^{\frac{\alpha-1}{2}}U_{E^{\prime}}^{i}U\left(  \sigma
^{\frac{1-\alpha}{2}}\rho^{\alpha}\sigma^{\frac{1-\alpha}{2}}\otimes
|0\rangle\langle0|_{E}\right)  U^{\dag}U_{E^{\prime}}^{i\dag}\left[
K^{i}\right]  ^{\frac{\alpha-1}{2}}\right]  ^{\frac{1}{\alpha}}\right\}
\label{eq:16}%
\end{align}
where%
\begin{equation}
K^{i}\equiv U_{E^{\prime}}^{i}U\left(  \sigma\otimes|0\rangle\langle
0|_{E}\right)  U^{\dag}U_{E^{\prime}}^{i\dag},
\end{equation}
with $\left\{  U_{E^{\prime}}^{i}\right\}  $ a set of Heisenberg-Weyl
operators. Then by the concavity result (\ref{eq:Hiai-concave}) above for
positive semidefinite $P$ and $Q$, it follows that the right-hand side of
(\ref{eq:16}) is no larger than%
\begin{multline}
\text{Tr}\left\{  \left[  \left[  \mathcal{N}(\sigma)\otimes\pi_{E^{\prime}%
}\right]  ^{\frac{\alpha-1}{2}}\left[  \mathcal{N}\left(  \sigma
^{\frac{1-\alpha}{2}}\rho^{\alpha}\sigma^{\frac{1-\alpha}{2}}\right)
\otimes\pi_{E^{\prime}}\right]  \left[  \mathcal{N}(\sigma)\otimes
\pi_{E^{\prime}}\right]  ^{\frac{\alpha-1}{2}}\right]  ^{\frac{1}{\alpha}%
}\right\} \\
=\text{Tr}\left\{  \left[  \mathcal{N}(\sigma)^{\frac{\alpha-1}{2}}%
\mathcal{N}\left(  \sigma^{\frac{1-\alpha}{2}}\rho^{\alpha}\sigma
^{\frac{1-\alpha}{2}}\right)  \mathcal{N}(\sigma)^{\frac{\alpha-1}{2}}\right]
^{\frac{1}{\alpha}}\right\}  .
\end{multline}
Thus we obtain the inequality%
\begin{equation}
\text{Tr}\left\{  \left[  \mathcal{N}(\sigma)^{\frac{\alpha-1}{2}}%
\mathcal{N}\left(  \sigma^{\frac{1-\alpha}{2}}\rho^{\alpha}\sigma
^{\frac{1-\alpha}{2}}\right)  \mathcal{N}(\sigma)^{\frac{\alpha-1}{2}}\right]
^{\frac{1}{\alpha}}\right\}  \geq1. \label{eq:helper-ineq}%
\end{equation}
This completes the proof of (\ref{delta-non-neg}) for $\alpha\in\left(
1,2\right)  $. The equality condition in (\ref{eq:Delta=0-other-a}) follows
from the representation in (\ref{eq:Delta-rewrite-1}) and the equality
condition stated after (\ref{eq:renyi-non-neg}).

Next we prove that $\widetilde{\Delta}_{\alpha}(\rho,\sigma,\mathcal{N})\geq0$
for $\alpha\in\left(  1,\infty\right)  $. We start with the definition
(\ref{eq:Renyi-rel-ent-diff-sand}), which we repeat here for convenience:%
\begin{equation}
\widetilde{\Delta}_{\alpha}(\rho,\sigma,\mathcal{N})\equiv\frac{\alpha}%
{\alpha-1}\log\left\Vert \rho^{\frac{1}{2}}\sigma^{\frac{1-\alpha}{2\alpha}%
}\mathcal{N}^{\dag}\left(  \mathcal{N}(\sigma)^{\frac{\alpha-1}{2\alpha}%
}\mathcal{N(}\rho)^{\frac{1-\alpha}{\alpha}}\mathcal{N(}\sigma)^{\frac
{\alpha-1}{2\alpha}}\right)  \sigma^{\frac{1-\alpha}{2\alpha}}\rho^{\frac
{1}{2}}\right\Vert _{\alpha}. \label{eq-sandwich-rel-diff}%
\end{equation}
From \cite[Lemma 12]{MDSFT13}, it follows that the right-hand side of
(\ref{eq-sandwich-rel-diff}) can be written as
\begin{multline}
\left\Vert \rho^{\frac{1}{2}}\sigma^{\frac{1-\alpha}{2\alpha}}\mathcal{N}%
^{\dag}\left(  \mathcal{N}(\sigma)^{\frac{\alpha-1}{2\alpha}}\mathcal{N(}%
\rho)^{\frac{1-\alpha}{\alpha}}\mathcal{N(}\sigma)^{\frac{\alpha-1}{2\alpha}%
}\right)  \sigma^{\frac{1-\alpha}{2\alpha}}\rho^{\frac{1}{2}}\right\Vert
_{\alpha}\\
=\sup_{\tau\geq0,\text{Tr}\left\{  \tau\right\}  \leq1}\text{Tr}\left\{
\rho^{\frac{1}{2}}\sigma^{\frac{1-\alpha}{2\alpha}}\mathcal{N}^{\dag}\left(
\mathcal{N}(\sigma)^{\frac{\alpha-1}{2\alpha}}\mathcal{N(}\rho)^{\frac
{1-\alpha}{\alpha}}\mathcal{N(}\sigma)^{\frac{\alpha-1}{2\alpha}}\right)
\sigma^{\frac{1-\alpha}{2\alpha}}\rho^{\frac{1}{2}}\tau^{\frac{\alpha
-1}{\alpha}}\right\}  .
\end{multline}
Now define%
\begin{equation}
\widetilde{\Delta}_{\alpha}(\rho,\sigma,\mathcal{N};\tau)\equiv\frac{1}%
{\alpha-1}\log\text{Tr}\left\{  \rho^{\frac{1}{2}}\sigma^{\frac{1-\alpha
}{2\alpha}}\mathcal{N}^{\dag}\left(  \mathcal{N}(\sigma)^{\frac{\alpha
-1}{2\alpha}}\mathcal{N}(\rho)^{\frac{1-\alpha}{\alpha}}\mathcal{N}%
(\sigma)^{\frac{\alpha-1}{2\alpha}}\right)  \sigma^{\frac{1-\alpha}{2\alpha}%
}\rho^{\frac{1}{2}}\tau^{\frac{\alpha-1}{\alpha}}\right\}  .
\end{equation}
Let us focus on the trace functional in the above equation:%
\begin{align}
&  \text{Tr}\left\{  \rho^{\frac{1}{2}}\sigma^{\frac{1-\alpha}{2\alpha}%
}\mathcal{N}^{\dag}\left(  \mathcal{N}(\sigma)^{\frac{\alpha-1}{2\alpha}%
}\mathcal{N}(\rho)^{\frac{1-\alpha}{\alpha}}\mathcal{N}(\sigma)^{\frac
{\alpha-1}{2\alpha}}\right)  \sigma^{\frac{1-\alpha}{2\alpha}}\rho^{\frac
{1}{2}}\tau^{\frac{\alpha-1}{\alpha}}\right\} \nonumber\\
&  =\text{Tr}\left\{  \sigma^{\frac{1-\alpha}{2\alpha}}\rho^{\frac{1}{2}}%
\tau^{\frac{\alpha-1}{\alpha}}\rho^{\frac{1}{2}}\sigma^{\frac{1-\alpha
}{2\alpha}}\mathcal{N}^{\dag}\left(  \mathcal{N}(\sigma)^{\frac{\alpha
-1}{2\alpha}}\mathcal{N}(\rho)^{\frac{1-\alpha}{\alpha}}\mathcal{N}%
(\sigma)^{\frac{\alpha-1}{2\alpha}}\right)  \right\} \nonumber\\
&  =\text{Tr}\left\{  \mathcal{N}\left(  \sigma^{\frac{1-\alpha}{2\alpha}}%
\rho^{\frac{1}{2}}\tau^{\frac{\alpha-1}{\alpha}}\rho^{\frac{1}{2}}%
\sigma^{\frac{1-\alpha}{2\alpha}}\right)  \mathcal{N}(\sigma)^{\frac{\alpha
-1}{2\alpha}}\mathcal{N}(\rho)^{\frac{1-\alpha}{\alpha}}\mathcal{N}%
(\sigma)^{\frac{\alpha-1}{2\alpha}}\right\} \nonumber\\
&  =\text{Tr}\left\{  \mathcal{N}(\sigma)^{\frac{\alpha-1}{2\alpha}%
}\mathcal{N}\left(  \sigma^{\frac{1-\alpha}{2\alpha}}\rho^{\frac{1}{2}}%
\tau^{\frac{\alpha-1}{\alpha}}\rho^{\frac{1}{2}}\sigma^{\frac{1-\alpha
}{2\alpha}}\right)  \mathcal{N}(\sigma)^{\frac{\alpha-1}{2\alpha}}%
\mathcal{N}(\rho)^{\frac{1-\alpha}{\alpha}}\right\}  .
\end{align}
By making the substitution $\alpha^{\prime}\equiv\left(  2\alpha-1\right)
/\alpha$, so that $\left(  1-\alpha\right)  /\alpha=1-\alpha^{\prime}$ and
thus $\alpha^{\prime}\in\left(  1,2\right)  $ when $\alpha\in\left(
1,\infty\right)  $, we see that the last line above is equal to%
\begin{equation}
\text{Tr}\left\{  \mathcal{N}(\sigma)^{\frac{\alpha^{\prime}-1}{2}}%
\mathcal{N}\left(  \sigma^{\frac{1-\alpha^{\prime}}{2}}\rho^{\frac{1}{2}}%
\tau^{\left(  \alpha^{\prime}-1\right)  }\rho^{\frac{1}{2}}\sigma
^{\frac{1-\alpha^{\prime}}{2}}\right)  \mathcal{N}(\sigma)^{\frac
{\alpha^{\prime}-1}{2}}\mathcal{N}(\rho)^{\left(  1-\alpha^{\prime}\right)
}\right\}  .
\end{equation}
Observe that this is similar to\ (\ref{eq:Delta-rewrite}). Hence we can write
$\widetilde{\Delta}_{\alpha}(\rho,\sigma,\mathcal{N};\tau)$ as%
\begin{multline}
\frac{1}{\alpha^{\prime}-1}\log\text{Tr}\left\{  \mathcal{N}(\sigma
)^{\frac{\alpha^{\prime}-1}{2}}\mathcal{N}\left(  \sigma^{\frac{1-\alpha
^{\prime}}{2}}\rho^{\frac{1}{2}}\tau^{\left(  \alpha^{\prime}-1\right)  }%
\rho^{\frac{1}{2}}\sigma^{\frac{1-\alpha^{\prime}}{2}}\right)  \mathcal{N}%
(\sigma)^{\frac{\alpha^{\prime}-1}{2}}\mathcal{N}(\rho)^{\left(
1-\alpha^{\prime}\right)  }\right\} \\
=D_{\alpha^{\prime}}\left(  \left[  \mathcal{N}(\sigma)^{\frac{\alpha^{\prime
}-1}{2}}\mathcal{N}\left(  \sigma^{\frac{1-\alpha^{\prime}}{2}}\rho^{\frac
{1}{2}}\tau^{\left(  \alpha^{\prime}-1\right)  }\rho^{\frac{1}{2}}%
\sigma^{\frac{1-\alpha^{\prime}}{2}}\right)  \mathcal{N}(\sigma)^{\frac
{\alpha^{\prime}-1}{2}}\right]  ^{\frac{1}{\alpha^{\prime}}}\middle\Vert
\mathcal{N}(\rho)\right)  .
\end{multline}
This implies that%
\begin{align}
&  \widetilde{\Delta}_{\alpha}(\rho,\sigma,\mathcal{N})\nonumber\\
&  =\sup_{\tau\geq0,\text{Tr}\left\{  \tau\right\}  \leq1}D_{\alpha^{\prime}%
}\left(  \left[  \mathcal{N}(\sigma)^{\frac{\alpha^{\prime}-1}{2}}%
\mathcal{N}\left(  \sigma^{\frac{1-\alpha^{\prime}}{2}}\rho^{\frac{1}{2}}%
\tau^{\alpha^{\prime}-1}\rho^{\frac{1}{2}}\sigma^{\frac{1-\alpha^{\prime}}{2}%
}\right)  \mathcal{N}(\sigma)^{\frac{\alpha^{\prime}-1}{2}}\right]  ^{\frac
{1}{\alpha^{\prime}}}\middle\Vert\mathcal{N}(\rho)\right) \\
&  \geq D_{\alpha^{\prime}}\left(  \left[  \mathcal{N}(\sigma)^{\frac
{\alpha^{\prime}-1}{2}}\mathcal{N}\left(  \sigma^{\frac{1-\alpha^{\prime}}{2}%
}\rho^{\frac{1}{2}}\rho^{\left(  \alpha^{\prime}-1\right)  }\rho^{\frac{1}{2}%
}\sigma^{\frac{1-\alpha^{\prime}}{2}}\right)  \mathcal{N}(\sigma
)^{\frac{\alpha^{\prime}-1}{2}}\right]  ^{\frac{1}{\alpha^{\prime}}%
}\middle\Vert\mathcal{N}(\rho)\right) \\
&  =D_{\alpha^{\prime}}\left(  \left[  \mathcal{N}(\sigma)^{\frac
{\alpha^{\prime}-1}{2}}\mathcal{N}\left(  \sigma^{\frac{1-\alpha^{\prime}}{2}%
}\rho^{\alpha^{\prime}}\sigma^{\frac{1-\alpha^{\prime}}{2}}\right)
\mathcal{N}(\sigma)^{\frac{\alpha^{\prime}-1}{2}}\right]  ^{\frac{1}%
{\alpha^{\prime}}}\middle\Vert\mathcal{N}(\rho)\right) \\
&  =\Delta_{\alpha^{\prime}}(\rho,\sigma,\mathcal{N})\\
&  \geq0,
\end{align}
where the first inequality follows by setting $\tau=\rho$ and the last
inequality follows because we have already shown that $\Delta_{\alpha^{\prime
}}(\rho,\sigma,\mathcal{N})\geq0$ for $\alpha^{\prime}\in\left(  1,2\right)
$. The above inequality also demonstrates that%
\begin{equation}
\widetilde{\Delta}_{\alpha}(\rho,\sigma,\mathcal{N})\geq\Delta_{\alpha
^{\prime}}(\rho,\sigma,\mathcal{N}). \label{eq:Delta-tilde-Delta}%
\end{equation}

\end{proof}

\bigskip

\noindent We can use this result to give an alternative proof of the fact that
$\mathcal{N}$ is sufficient for $\rho$ and $\sigma$ if $\Delta_{\alpha}%
(\rho,\sigma,\mathcal{N})=0$ for some $\alpha\in(1,2)$ or if $\widetilde
{\Delta}_{\alpha}(\rho,\sigma,\mathcal{N})=0$ for some $\alpha\in\left(
1,\infty\right)  $:

\bigskip

\begin{proof}
If for some $\alpha\in(1,2)$ we have $\Delta_{\alpha}(\rho,\sigma
,\mathcal{N})=0$, then we know that%
\begin{align}
\left[  \mathcal{N}(\sigma)^{\frac{\alpha-1}{2}}\mathcal{N}\left(
\sigma^{\frac{1-\alpha}{2}}\rho^{\alpha}\sigma^{\frac{1-\alpha}{2}}\right)
\mathcal{N}(\sigma)^{\frac{\alpha-1}{2}}\right]  ^{\frac{1}{\alpha}}  &
=\mathcal{N}(\rho)\\
\Longleftrightarrow\mathcal{N}\left(  \sigma^{\frac{1-\alpha}{2}}\rho^{\alpha
}\sigma^{\frac{1-\alpha}{2}}\right)   &  =\mathcal{N}(\sigma)^{\frac{1-\alpha
}{2}}\mathcal{N}(\rho)^{\alpha}\mathcal{N}(\sigma)^{\frac{1-\alpha}{2}}%
\end{align}
This implies that%
\begin{equation}
\text{Tr}\left\{  \mathcal{N}\left(  \sigma^{\frac{1-\alpha}{2}}\rho^{\alpha
}\sigma^{\frac{1-\alpha}{2}}\right)  \right\}  =\text{Tr}\left\{
\mathcal{N}(\rho)^{\alpha}\mathcal{N}(\sigma)^{1-\alpha}\right\}
\end{equation}
which in turn implies that%
\begin{equation}
\text{Tr}\left\{  \rho^{\alpha}\sigma^{1-\alpha}\right\}  =\text{Tr}\left\{
\mathcal{N}(\rho)^{\alpha}\mathcal{N}(\sigma)^{1-\alpha}\right\}  ,
\end{equation}
because $\mathcal{N}$ is trace preserving. This is then an equality of
$f$-divergences, which we can use to conclude the sufficiency property as in
previous proofs.

If $\widetilde{\Delta}_{\alpha}(\rho,\sigma,\mathcal{N})=0$, for some
$\alpha\in\left(  1,\infty\right)  $, then from (\ref{eq:Delta-tilde-Delta})
we have that $\Delta_{\alpha^{\prime}}(\rho,\sigma,\mathcal{N})=0$ for some
$\alpha^{\prime}\in(1,2)$. Then we know from the above analysis that
$\mathcal{N}$ is sufficient for $\rho$ and $\sigma$.
\end{proof}

\bibliographystyle{alpha}
\bibliography{Ref}

\end{document}